\documentclass[10pt, journal,fleqn]{IEEEtran}  

\IEEEoverridecommandlockouts     

\usepackage{color,amsmath,amsthm,amstext,amsfonts,amssymb, mathrsfs,mathtools}
\usepackage{graphicx,algorithm,algorithmic}
\usepackage{tikz}
\usepackage{setspace}
\usepackage{array} 
\usepackage{booktabs}  
\usepackage{bbm}
\usepackage{wasysym}
\usepackage{textcomp}
\usepackage{hyperref} 
\usepackage[sort,compress]{cite}

\newtheorem{proposition}{Proposition}
\newtheorem{theorem}{Theorem}
\newtheorem{definition}{Definition}

\theoremstyle{remark}
\newtheorem{remark}{Remark} 
\newtheorem{example}{Example}

\newcommand{\Ccal}{\cal C}

\newcommand{\Eset}{\mathbb{E}}

\newcommand{\Hset}{\mathbb{H}}
\newcommand{\Pset}{\mathbb{P}}
\newcommand{\Sset}{\mathbb{S}}

\newcounter{l1}
\newcounter{l2}
\newcounter{l3}
\newcommand{\bdotlist}{\begin{list}{$\bullet$}{}}
\newcommand{\bboxlist}{\begin{list}{$\Box$}{}}
\newcommand{\bbboxlist}{\begin{list}{\raisebox{.005in}{{\tiny $\blacksquare$ \ \ }}}{}}
\newcommand{\bdashlist}{\begin{list}{$-$}{} }
\newcommand{\blist}{\begin{list}{}{} }
\newcommand{\barablist}{\begin{list}{\arabic{l1}}{\usecounter{l1}}}
\newcommand{\balphlist}{\begin{list}{(\alph{l2})}{\usecounter{l2}}}
\newcommand{\bAlphlist}{\begin{list}{\Alph{l2}.}{\usecounter{l2}}}
\newcommand{\bdiamlist}{\begin{list}{$\diamond$}{}}
\newcommand{\bromalist}{\begin{list}{(\roman{l3})}{\usecounter{l3}}}

\newcommand{\beq}{\begin{equation}}
\newcommand{\eeq}{\end{equation}}

\newcommand{\dst}[1]{\displaystyle{ #1 }}

\newcommand{\tn}{\textnormal}
\DeclarePairedDelimiter{\ceil}{\lceil}{\rceil}
\title{Mechanism Design for \\  Demand Response Programs}

\author{Deepan Muthirayan, Dileep Kalathil, Kameshwar Poolla, Pravin Varaiya 
\thanks{Deepan Muthirayan is with the School of Electrical Engineering and Computer Science, UC Irvine. Dileep Kalathil is the Department of Electrical and Computer Engineering, Texas A\&M University, College Station. Kameshwar Poolla and Pravin Varaiya are with the Department of Electrical Engineering and Computer Sciences, University of California, Berkeley.}
\thanks{$^\pi$Supported in part by the National Science Foundation under Grants EECS-1129061/9001, 
CPS-1239178; and by CERTS under sub-award 09-206; PSERC S-52.}
}

\begin{document}
\maketitle
\markboth{Submission for IEEE Transactions on Smart Grids} 
{Muthirayan \emph{et al.}}

\begin{abstract}
Demand Response (DR) programs serve to reduce the consumption of electricity at times when the supply is scarce and expensive. 
Consumers or {\em agents} with flexible consumption profiles are recruited by an aggregator who manages the DR program.
The aggregator calls on a subset of its pool of recruited agents to reduce their electricity use during DR events.
Agents are paid for reducing their energy consumption from contractually established baselines.  
{\em Baselines} are counter-factual consumption estimates of the energy an agent would have consumed if they 
were not participating in the DR program. Baselines are used to determine payments to agents. 
This creates an incentive for agents to inflate their baselines in order to increase the payments they receive. 
There are several newsworthy cases of agents gaming their baseline for economic benefit. We propose a novel {\em self-reported} baseline mechanism (SRBM) where each agent reports its baseline and marginal utility. These reports are strategic and need not be truthful. Based on the reported information, the aggregator selects or {\em calls} on agents with a certain probability to meet the load reduction target $D$.  Called agents are paid for observed reductions from their self-reported baselines. Agents who are not called face penalties for consumption shortfalls below their baselines.
Under SRBM, we show that truthful reporting of baseline consumption and marginal utility is a dominant strategy. Thus, SRBM
eliminates the incentive for agents to inflate baselines. SRBM is assured to meet the load reduction target.  
Finally, we show that SRBM is almost optimal in the metric of average cost of DR provision faced by the aggregator.
\end{abstract}

\vspace*{-0.1in}

\section{Introduction} 

The core problem in power systems operations is to maintain the fine balance of electricity supply and demand at all times. 
This must be done economically through markets while respecting resource and reliability constraints. 
Adeptly managing flexible demand is a far better alternative to increased reserve generation, 
since it is inexpensive, produces no emissions, and consumes no resources. 
While most DR programs are limited to infrequent peak shaving applications, it is recognized that 
demand flexibility has the potential to offer more lucrative ancillary services such as frequency regulation or load-following. 
These applications can support balancing supply and demand to compensate for the variability of renewables.  
This paper, however, is concerned only with peak shaving DR applications.   

In DR programs, aggregators recruit  residential or industrial customers who are willing to reduce their electricity consumption 
in exchange for financial rewards. The aggregator serves as an intermediary and represents these flexible consumers  
or {\em agents} to the local utility. The aggregator receives a  payment from the utility for the ability to reduce demand at short notice, 
and,  in-turn, pays the agents for  their consumption reduction during DR events. The key difficulty is in measuring 
this reduction in consumption. While the actual consumption of agents is measured, their intended consumption or {\em baseline} is a counter-factual.

Commonly used baselines include historical averages of consumption on similar days (by the agent, or by a peer group of similar agents). 
However, there are newsworthy cases where agents have deliberately inflated their baseline to extract larger payments \cite{gaming-examples}. 
Inaccurate baselines can result in  over-payment, compromising the cost-effectiveness of the DR program, or in  
under-payment, adversely affecting the ability to recruit participants into DR programs. 
Finally, fairness can be of concern. An agent who happens to be on vacation during a DR event receives a payment for load reduction 
without suffering any hardship. This can be perceived as unfair by other agents who deliberately curtail their consumption 
and suffer some dis-utility. Addressing these issues is essential to encourage and sustain wider use of DR programs.

\subsection{Our Contributions} 

We approach DR program planning as a {\em mechanism design} problem. We propose a novel {\em self-reported} baseline mechanism (SRBM) where
agents {\em self-report} baselines which are forecasts of their intended future consumption,
and their marginal utilities to an aggregator who manages the DR program. Agents need not be truthful in their reporting. 
The self-reported baseline is used to determine payments
in the DR program. 
The objective of the aggregator is to design an incentive mechanism so that (a) the DR program delivers any load reduction target, 
(b) each agent reports their true baseline and true marginal utility, and 
(c) the DR program delivers the load reduction target at minimum cost to the aggregator. 

Under SRBM, we show that  truthful reporting of  baseline consumption and marginal utility
is a dominant strategy for each agent. 
We show that agents are faithful to their baseline  consumption when not called for DR, 
and that agents maximally reduce their consumption when called for DR. Under this mechanism, 
the aggregator can ensure adequate response, i.e. an assurance of being able to
deliver the load reduction target $D$. 
We characterize the minimum possible average cost $\phi_\text{min}$ 
of DR provision per KWh under a class of mechanisms. We then show that SRBM is nearly optimal in the metric of
expected cost of DR provision, i.e. it results in a cost close to $\phi_\text{min}$.

There is extensive literature on network resource allocation problems and all of them consider an infinitely divisible good or service that is to be efficiently shared among distributed agents acting with self interests  \cite{semret1999market, zou2018resource}. Our setting also considers an infinitely divisible service request and selfish service providers (DR providers) but the difference being that the service provided is not measurable because of the lack of baseline against which DR service provided has to be measured. This is a challenge for mechanism design. Moreover, the literature in network resource allocation problems \cite{semret1999market, zou2018resource} consider multishot setting and study convergence properties to a Nash Equilibrium over repeated reporting, allocation and pricing. The few works which consider a single shot setting require the consumers to report a very high dimensional bid for convergence in a single shot. Here we study convergence in a single shot with an additional challenge of lack of baseline measurement. Therefore our contribution is proposing a resource allocation dominant strategy mechanism in a single shot setting with this new challenge.

\subsection{Related Work}
Traditional DR programs reward participating agents for load reduction during peak demand periods. 
Agents have an incentive to deliberately inflate their baselines to increase the payments they receive \cite{chao2010price, chao2013incentive, bushnell2009comes, wolak2007residential}. Alternative baseline mechanisms (ex: aggregate baselines) which 
improve market efficiency are offered in \cite{chao2013incentive}.  These do not explicitly address baseline inflation concerns. 
Adverse selection and double payment effects are two other issues that arise from  rewarding agents based on 
estimated baselines \cite{chao2010price}. 



There is a substantial literature on baseline estimation methods. These can be broadly classified into (a) averaging, (b) regression, and (c) control group methods. 

Averaging methods determine baselines by averaging the consumption on past days that are similar 
(ex: in temperature or workday) to the event day. 
There are many variants such as weighted averaging and using an adjustment factor to account 
for variations between the event day and prior similar days.    
A detailed comparison of  different averaging methods in offered in \cite{coughlin2008estimating, grimm2008evaluating, wijaya2014bias}. 
While averaging methods are attractive because of their simplicity, they suffer from estimation biases that can be 
substantial \cite{wijaya2014bias, weng2015probabilistic}. Also, these methods require significant data access, 
especially for residential DR programs \cite{nolan2015challenges}. 

Regression methods fit a load prediction model to historical data, which is then used to predict the baseline \cite{zhou2016forecast, mathieu2011quantifying}. 
They can potentially overcome biases incurred by averaging methods \cite{nolan2015challenges}. 
They often require considerable historical data for acceptable accuracy, and the models may be  
too simple to capture the complex behavior of individual agents.   

Control group methods are found to be more accurate 
than averaging or regression methods \cite{hatton2016statistical}. 
While they do not require large amount of historical data, they require additional metering infrastructure.
This complicates and raises the costs of implementation, particularly for large numbers of recruited agents. 
Finally, \cite{weng2015probabilistic} proposes a probabilistic method using Gaussian statistics to estimate baselines. We refer the readers to \cite{xia2017energycoupon} for a discussion on these different baseline schemes. 
 
Most of the methods mentioned above focus on baseline estimation. They often overlook the behavioral or gaming aspects of agents intentionally inflating their baselines. There are some exceptions, notably \cite{chen2012cheat} which considers a linear penalty when agents deviate from their baselines and shows that the penalty induces users to report their true baselines. In \cite{ramos2013asymmetry}, again under linear penalties, a centralized DR scheduling algorithm is proposed to guarantee incentive compatibility in the case of two agents. A DR market assuming known baselines is proposed in \cite{nguyen2013market, nguyen2011pool}.
The approaches in \cite{chen2012cheat, ramos2013asymmetry, nguyen2013market, nguyen2011pool} either assume knowledge of utility functions or true baselines. In \cite{dobakhshari2016contract} authors propose an optimal contract mechanism for the DR aggregator. However they assume that true reduction can be observed at a later time, and the payment depends on this information. Our approach doesn't require this assumption. Instead we propose a joint design of baseline estimation and incentive design to address both problems together.

\section{Problem Formulation}  
\label{sec:preliminaries}

\vspace{-0.5em}
{\small
\begin{table}
\caption{Notations}
\begin{tabular}{c|l} 
\toprule
$\Eset[ X ]$ & expected value of the random variable $X$ \\
$D$ & load reduction target \\
$m$ & number of DR events agents must participate in \\
$N$ & number of agents recruited by aggregator\\
$u_k$ & utility of agent $k$ \\
$q_{k}$ & discretionary energy consumption of agent $k$\\
$b_k$ & true baseline consumption of agent $k$\\
$\pi_k$ & true marginal utility of agent $k$\\
$\pi_\text{max}$ & upper bound on marginal utilities, $\pi_\text{max} \geq \max_k \pi_k$ \\
$\alpha_k$ & probability that agent $k$ is selected \\
$f_k$ & baseline report of agent $k$\\
$\mu_k$ & marginal utility report of agent $k$\\ 
$\pi^r_k$ & reward/kWh awarded to agent $k$\\
$\pi^p_k$ & penalty/kWh imposed on agent $k$\\
$\pi^e$ & retail price of energy \\
$\pi^o$ & recruitment cost per enrolled agent \\
$\Pset^i$ & pod $i$ \\
$\Sset^i$ & pod core $i$ \\
$\Hset^i$ & pod header $i$ \\
$\beta^i$ & probability that pod $i$ is selected for DR \\
$\nu^i$ & maximum reported marginal utility in pod $i$ \\
$\phi$ & average cost of DR provision per KWh \\
$\psi$ & payout to agents per DR event \\
\bottomrule
\end{tabular}
\end{table}}
\label{tab:notation}

\subsection{Consumer Model}
Let $u_{k}(q_{k})$ be the utility of agent $k$ derived by consuming $q_k$ units of energy. We assume that the utility functions 
$u_k(\cdot)$ have the piece-wise linear form
\begin{equation}
\label{eq:utility-1}
u_{k}(q_{k}) = \left\lbrace \begin{array}{cc}
\pi_{k} q_{k} & \tn{if} ~ q_{k} < b_k \\
\pi_{k} b_k & \tn{if} ~ q_{k} \geq b_k
\end{array} \right.
\end{equation}
Here $b_k$ is the maximum possible consumption of agent $k$. Any additional consumption will not increase its utility. 
We call $\pi_{k}$ the {\em true marginal utility} of agent $k$. We assume that $\pi_k$s are i.i.d and so are $b_k$s and that $\pi_k$ and $b_k$s are independent of each other.

Let $\pi^{e}$ be the retail price of electricity
offered by the utility.  
The net utility $U_{k}(\cdot)$ of agent $k$ is
\begin{equation*}
U_{k}(q_{k}) = \left\lbrace \begin{array}{cc}
\pi_{k}q_{k} -\pi^{e}q_{k} & \tn{if} ~q_{k} < b_k \\
\pi_{k}b_k -\pi^{e}q_{k} & \tn{if} ~q_{k} \geq b_k
\end{array} \right.  
\end{equation*} 
We clearly require $\pi_{k} > \pi^{e}$, else agent $k$ would not consume electricity. Equivalently, for every agent, the marginal utility derived from electricity consumption exceeds the
retail electricity price.

\begin{figure}[h!]
\centering
\begin{tikzpicture}[scale = 0.8]

\draw [->, thick] (0,0) -- (4,0);
\draw (3.75,-0.2) node[anchor=north] {\small $q_k$};
\draw (2,-0.2) node[anchor=north] {\small $b_k$};
\draw [->, thick] (0,0) -- (0,3);
\draw (-0.2,2.75) node[anchor=east] {\small $u_k$} ;
\draw[blue, thick] (0,0) -- (2,2.25) -- (4,2.25);
\draw[dashed] (2,0) -- (2,3);
\draw[dashed] (0,2.25) -- (2,2.25);
\draw [red,thick] (0.8,0.9) -- (1.2,0.9) -- (1.2, 1.35);
\draw (1.25,1.1) node[anchor=west] {\small $\pi_k$};

\begin{scope}[xshift = 2in]
\draw [->, thick] (0,0) -- (4,0);
\draw (3.75,-0.2) node[anchor=north] {\small $q_k$};
\draw (2,-0.2) node[anchor=north] {\small $b_k$};
\draw [->, thick] (0,0) -- (0,3);
\draw (-0.2,2.75) node[anchor=east] {\small $U_k$} ;
\draw[blue, thick] (0,0) -- (2,2) -- (4,1.5);
\draw[dashed] (2,0) -- (2,3);
\draw[dashed] (0,2) -- (4,2);
\draw [red,thick] (0.8,0.8) -- (1.2,0.8) -- (1.2, 1.2);
\draw (.35,.35) node[anchor=west] {\small $\pi_k - \pi^e$};
\draw [red,thick] (2.8,1.8) -- (2.8,1.6) -- (3.6, 1.6);
\draw (2.7,1.3) node[anchor=west] {\small $- \pi^e$};
\end{scope}

\end{tikzpicture}%
\vspace*{-0.1in}
\caption{(a) Utility, and (b) net utility of agent $k$}
\label{fig:conut}
\end{figure}
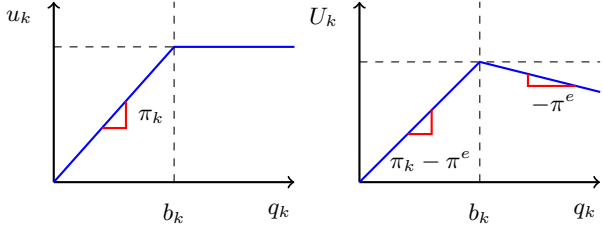

The optimal consumption for agent $k$  maximizes the net utility. This is $b_k$ as is evident from Figure \ref{fig:conut}(b).
We call $b_k$ the {\em true baseline consumption} of agent $k$. 

\begin{remark}
We are only modeling the {\em discretionary} electricity consumption of an agent. 
This is the consumption that the agent is willing to forgo
in exchange for monetary compensation. With this interpretation, $b_k$ is the {\em maximum consumption reduction}
that agent $k$ is willing to provide, and $\pi_k$ is the {\em minimum compensation} per KWh it requires to provide this reduction.
\end{remark}

\begin{remark}
The agent utility functions are private information. The aggregator does not have knowledge of agent baselines and marginal utilities.
We assume that the aggregator has knowledge of an upper bound on the agent marginal utilities, i.e. it knows $\pi_\text{max}$ where
\begin{equation} \pi_\text{max} \geq \max_{k} \pi_{k}. \end{equation}
We note that $\pi_\text{max}$ has the interpretation of the maximum price that the aggregator is {\em willing to pay} agents per KWh of demand reduction.
\end{remark}

\begin{remark} 
We assume that the agent parameters $\pi_k$ and $b_k$ are independent across the agents. This is reasonable because the true marginal utility and the true baseline of an agent are not dependent on the behavior of other agents.
\end{remark}

\begin{remark}
The agent's utility function is chosen to be a piecewise linear concave function. This simplifies the problem setting and allows a more tractable analysis while retaining all the challenges that the general problem poses. If the agent's utility function is modeled by a general concave function, as in network resource allocation problems, then for convergence in single shot this would necessarily require the mechanism designer to elicit very high dimensional bids. This, in addition to the lack of baseline measurement, makes the general problem complex for analysis. We suggest this as a future direction for research.
\end{remark}

\subsection{Aggregator Model}

The utility is financially motivated to reduce its procurement costs at times when supply is scarce and expensive.  
It purchases demand response services from aggregators in capacity markets.
Suppose the aggregator commits to offer the utility $D$ KWh of demand reduction during DR events over some contract window.
It recruits $N$ agents into its DR program from a large candidate pool. We will see later that $N$ is a function of the consumer {\it selection mechanism}.
The cost of recruitment is $\pi^o$ per enrolled agent.
{\it The recruited agents are obligated to participate in $m$ DR events contractually}.

The aggregator profit is the revenue from the utility, minus the payout to the agents and recruiting costs. 
It may also receive penalty revenue from agents, but we will show that this is not the case under our baseline mechanisms.
The total expected cost faced by the aggregator is 
\[ J_\text{agg} =  m \Eset[ \psi ] +  \pi^o \Eset [ N  ] \]
where $\psi$ is the payout per DR event, $N$ is the number of recruited agents and expectation is over consumer baseline and marginal utility.
The aggregator's {\em expected cost of demand response} $\phi$, i.e. the average cost per KWh of demand reduction is then
\beq 
\label{eq:aggcost} \phi = \frac{J_\text{agg}}{D} = \underbrace{\frac{\Eset [ \psi ]}{D}}_\text{payout per KWh} + \underbrace{\frac{\pi^o\Eset[N]}{mD}}_\text{recruitment cost} \eeq


\subsection{Event Time-line}  
\label{sec:event-timeline}
\vspace{-0.2cm}
The time-line of events in our problem formulation is shown in Figure \ref{fig:timeline}. 
We divide these events into four periods.

\begin{figure}[h!]
\centering
\begin{tikzpicture}[scale = 0.9]

\draw [->, very thick] (0,0) -- (9.5,0);
\draw[dashed] (3,-1) -- (3,1.75);
\draw[dashed] (6,-1) -- (6,1.75);
\draw (1.5,1.75) node[anchor=north] {\small Period 1};
\draw (4.5,1.75) node[anchor=north] {\small Period 2};
\draw (7.5,1.75) node[anchor=north] {\small Period 3};

\draw (1.5,0.5) node[align=left] {\scriptsize utility notifies \\[-0.05in] \scriptsize agg of DR event};
\draw[thick,->] (0.3,0) -- (0.3,0.5);
\draw (1.4,-0.5) node[align=right] {\scriptsize agents report \\[-0.05in] \scriptsize private info $f_k, \mu_k$ };
\draw[thick,->] (2.7,0) -- (2.7,-0.5);

\draw (4.6,0.5) node[align=left] {\scriptsize agg selects  \\[-0.05in] \scriptsize agents for reduction};
\draw[thick,->] (3.3,0) -- (3.3,0.5);
\draw (4.25,-0.5) node[align=right] {\scriptsize agg publishes \\[-0.05in] \scriptsize penalty/reward prices };
\draw[thick,->] (5.7,0) -- (5.7,-0.5);

\draw (7.5,0.5) node[align=left] {\scriptsize agents decide  \\[-0.05in] \scriptsize consumption $q_k$};
\draw[thick,->] (6.3,0) -- (6.3,0.5);
\draw (7.5,-0.5) node[align=right] {\scriptsize agents pay penalties \\[-0.05in] \scriptsize or  recieve rewards};
\draw[thick,->] (8.9,0) -- (8.9,-0.5);

%

\end{tikzpicture}%
\caption{Event Time-line}
\label{fig:timeline}
\end{figure}
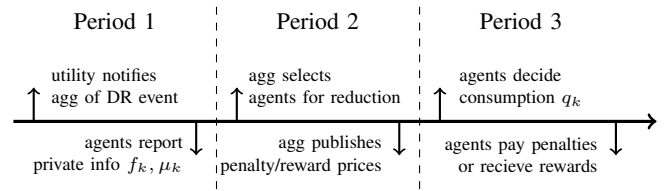

\textit{Period 0 (Common information):} The aggregator recruits $N$ agents  into its DR program.
Agents enroll based on the opportunity to receive financial rewards.  The aggregator informs participating agents 
of (a) agent selection mechanism, and (b) the mechanism which sets penalty and reward prices. 

\textit{Period 1 (Reporting):} The utility notifies the aggregator of an anticipated DR event during which it is 
obligated to deliver $D$ KWh of demand reduction.  
All agents (indexed by $k$) report their baselines and marginal utilities, $f_k$ and $\mu_k$ respectively,  to the aggregator. 
Agents need not be truthful. We stress that the agent reports $(f_k, \mu_k)$ are strategic, 
i.e. agents may opt to deliberately submit incorrect reports.

\textit{Period 2 (Selection):} The aggregator delivers the aggregate load reduction target $D$ by selecting 
a subset of agents to call on for consumption reduction. This selection is based on the collective reports submitted by the agents. 
The aggregator notifies selected agents to reduce their consumption. The aggregator computes and publishes  reward prices 
$\pi^r_{k}$ for agents who are called,  and  a penalty price $\pi^p_k$ for 
agents who are not called.  These can be  agent-specific.

\textit{Period 3 (Load reduction and payment):} 
During the DR event, all agents decide on their actual consumption $q_k$.
If agent $k$ is called, it receives an {\em ex post}  reward 
\beq 
R = \pi_k^r (f_k - q_k)^+.
\label{eq:reward1}
\eeq
If agent $k$ is not called, it is assessed an {\em ex post} penalty 
\beq 
P = \pi^p_k (f_k - q_k)^+ .
\label{eq:penalty1}
\eeq
Thus, selected agents are rewarded for consumption reduction from their reported baselines, and agents that are not selected are
penalized for consumption shortfalls below their reported baselines.

%

\subsection{The Agent's Problem} 
We assume that the agents are rational and non-cooperative. Each agent faces a two stage decision problem. 
In the first stage, it has to decide the value of its reports $(f_k, \mu_k)$. In the second stage, it has to decide on its
actual  energy consumption $q_{k}$ during the DR event. 
This second stage decision depends on whether or not agent $k$ is called for the DR event. 

There are two serious complexities that arise. First, agents are selected at {\em random} by the aggregator in response to a DR event.
The selection probability depends on the details of the aggregator's mechanism, which, in turn, depends on the submitted reports of {\em all} agents. 
Second, the reward and penalty prices faced by agent $k$ are set by the aggregator. These vary by agent and depend 
on the collective reports of the other agents $i \neq k$. 

We could approach the agent's decisions through a complicated game-theoretic formulation, but this is unnecessary.
It happens that for our self reported baseline mechanism (detailed in Section \ref{sec:mu-unknown}), {\em truthful reporting is the dominant strategy}.
Equivalently, agent $k$ will choose to reveal its true baseline and true marginal cost in the first stage. This results in the lowest expected cost for agent $k$, 
regardless of the reports of other agents (see Theorem \ref{thm:muuc-main}).

\section{Baseline-only Reporting} \label{sec:baselineonly}

Before we describe our self-reported baseline mechanism, we consider a simpler scheme where agents 
(a) only report their baselines, and (b) face
uniform reward and penalty prices independent of their reports, set by the aggregator.
This intermediate analysis will inspire our more complex mechanism where agents report both their baselines and 
marginal utilities, and face nonuniform prices.
  
\subsection{Mechanism Definition}

The aggregator recruits $N$ agents who are contractually obligated to participate in $m$ DR events.
 In response to  a DR event, all  recruited agents are required to only report their baselines $f_k$. 
The decisions of agents are not coupled as the penalty and reward prices $\pi^r, \pi^p$ are constants chosen by the aggregator.
The aggregator selects agents independently with probability $\alpha$ until
\beq  
\sum_k f_k > D.
\label{eq:aa} 
\eeq
Meeting (\ref{eq:aa}) requires that a sufficient number $N$ of agents be recruited. We will explore this aspect of the mechanism in 
Theorem \ref{thm:aa}. 

The aggregator's choices under this mechanism are $(N, \pi^r, \pi^p, \alpha)$.
The decision variables for agent  $k$ are its reported baseline $f_k$ and second stage consumption $q_k$ during the DR event.

\subsection{Agent Decisions}

The objective of each agent is to minimize its expected cost.
Agent $k$ solves a two-stage optimization problem. In the first stage, it optimally selects its baseline report $f_k$.
In the second stage, it is informed whether or not it is selected, and then optimally selects its consumption $q_k$ during the DR event.  

Suppose agent $k$ submits a baseline report $f_k$. If this agent is selected, its second stage cost function is
\[ J_\text{s}(q_k, f_k) = \pi^e q_k - u(q_k) - \pi^r (f_k - q_k)^+. \]
If the agent is not selected, it second stage cost function is
\[ J_\text{ns}(q_k, f_k) = \pi^e q_k - u(q_k) + \pi^p (f_k - q_k)^+. \]

Define the optimal consumptions
\begin{eqnarray*}
q^*_\text{s} & = & \arg\min_{q_k} J_\text{s}(q_k, f_k),~~ q^*_\text{ns} =  \arg\min_{q_k} J_\text{ns}(q_k , f_k). 
\end{eqnarray*}
Note that these depend on the first stage decision $f_k$, i.e. the reported baselines.

Let $\alpha$ be the selection probability. The first stage decision problem of agent $k$ is to minimize its expected cost:
\[ J(f_k) =  \alpha J_\text{s}(q^*_\text{s}, f_k) + (1 - \alpha) J_\text{ns}(q_\text{ns}^*, f_k).  \]
Let $f_k^*$ be the optimal first stage decision of agent $k$, i.e. its reported baseline:
\[ f_k^* = \arg\min J(f_k). \]

We have the following result.
\begin{theorem} \label{thm:1}
Suppose the mechanism prices satisfy
$\pi^r \geq \pi_\text{max} - \pi^e, \quad \pi^p \geq \pi^e$, 
and the probability of an agent being called satisfies $\alpha \leq \pi^e/(\pi^r + \pi^e)$. 
Under the baseline-only reporting mechanism, 
\setlength{\leftmargini}{0.15in}
\balphlist \setlength{\itemsep}{-0.1cm}
\item agents truthfully report their baselines, i.e. $f_k^* = b_k$ 
\item called agents consume $q_\text{s}^* = 0$ 
\item agents that are not called consume $q_\text{ns}^* = b$
\item the aggregator receives no penalty revenue 
\item the load reduction target $D$ is met.
\end{list}
\end{theorem}
{\em Proof:} Refer Appendix 

\begin{remark}
Result (a) assures us that there is no baseline inflation. Result (b) states that called agents maximally reduce their discretionary consumption.
Result (c) implies that agents who are not called consume the same amount of electricity as they would have if 
they were not participating in the DR program.
\end{remark}

\subsection{Aggregator Cost}

We now examine the aggregator's perspective. Consider the class $\Ccal$ of {\em self-reporting baseline mechanisms} with {\em linear} reward and penalty functions as specified in \eqref{eq:reward1} and \eqref{eq:penalty1}. The reward prices can be agent-specific.
Our next result offers a lower bound on the {\em minimum possible cost} per KWh of DR provision under {\em any mechanism in $\Ccal$}.
\begin{theorem} \label{thm:min}
Suppose agent true baselines $\{b_1, b_2, \cdots \}$ and agent true marginal utilities $\{\pi_1, \pi_2, \cdots\}$ 
are i.i.d. random variables. Assume $b_{i}$s and $\pi_{k}$s are independent for all $i, k$. 
Let  $\Ccal$ denote the class of self-reporting baseline mechanisms with linear reward and penalty functions.
Then, the expected cost for DR provision 
under any mechanism in $\Ccal$ satisfies 
\beq \label{eq:phimin}
 \phi \geq \frac{1}{\Eset[ 1/\pi] } - \pi^e + \frac{\pi^o}{m \pi^e \Eset [b] \Eset [ 1/\pi]} = \phi_\text{min}
 \eeq
 where $\Eset [b] = \Eset [b_{i}]$ and $\Eset [ 1/\pi] = \Eset [ 1/\pi_{k}]$. 
\end{theorem}
{\em Proof:} Refer Appendix

\begin{remark}
Note that the marginal utilities are bounded and bounded away from zero as $\pi^e \leq \pi_k \leq \pi_\text{max}$. 
A standard calculation reveals that
\[ 1/\Eset[ \pi] \leq \Eset[ 1/\pi] \leq 1/\Eset[ \pi] + {\mathcal O} (\sigma^2/m^3) \]
where $(m,\sigma^2)$ are the mean and variance of $\pi$. Thus, if the variance of true marginal utilities is modest, 
we can approximate $\Eset [ 1/\pi] \approx 1/ \Eset [\pi]$, and the minimum 
expected cost of DR provision over any mechanism in $\Ccal$ is
\beq 
\label{eq:802} 
\phi_\text{min} \approx \Eset[\pi] - \pi^e + \frac{\pi^o \Eset [\pi ]}{m \pi^e \Eset [b]} 
\eeq
\end{remark}

We next compute the minimal cost of DR under the baseline-only reporting mechanism. 
Minimizing this cost requires the aggregator to select the smallest possible reward price $\pi^r$, 
and the fewest number of customers $N$ (or the largest selection probability $\alpha$).
We have the following:

\begin{theorem} \label{thm:aa}
Suppose agent true baselines $\{b_1, b_2, \cdots \}$ and agent true marginal utilities $\{\pi_1, \pi_2, \cdots\}$ 
are i.i.d. random variables. Assume $b_{i}$s and $\pi_{k}$s are independent for all $i, k$. Under baseline-only reporting,
the aggregator's profit is maximized by the optimal parameters:
\setlength{\leftmargini}{0.15in}
\balphlist \setlength{\itemsep}{-0.1cm}
\item reward price: $\dst{\pi^r = \pi_\text{max} - \pi^e}$
\item penalty price: $\dst{\pi^p \geq \pi^e}$
\item selection probability: $\dst{\alpha  = \pi^e / \pi_\text{max}}$
\end{list}
The resulting expected cost per KWh of DR provision is
\beq \label{eq:phiBO}
 \phi_\text{BO}  \leq (\pi_\text{max} - \pi^e)  \left(1 + \frac{ \pi^e \Eset[ b ] }{\pi_\text{max}D } \right)   + 
\frac{\pi^o \pi_\text{max} }{m \pi^e \Eset[ b ] }  + \frac{\pi^o}{mD}
\eeq
\end{theorem}

{\em Proof:} Refer Appendix

\begin{remark} The penalty price does not affect the aggregator's cost as it derives no penalty revenue. The only constraint is that
$\pi^p \geq \pi^e$ in order for agents to report truthfully. 
\end{remark}

\begin{remark}
For large demand reduction targets $D$, our upper bound (\ref{eq:phiBO}) on the average cost of DR provision becomes
\beq \label{eq:801} \phi_\text{BO}  \leq \pi_\text{max} - \pi^e   +  \frac{\pi^o \pi_\text{max} }{m \pi^e \Eset[ b ] }  \eeq

Comparing this with (\ref{eq:802}), we see that the baseline-only reporting mechanism incurs a large cost of DR provision 
because the aggregator does not have information about the true marginal utilities of agents. It only has access to the upper bound $\pi_\text{max}$.  Cost increase stems from the inflation of $\pi_\text{max}$ over $\Eset [ \pi ]$.
\end{remark}

\begin{example} \label{ex:1}
Assume that the marginal utilities of recruited agents are uniformly distributed on
$[0.3,1.3]$ \$/KWh. Equivalently, the mix of recruited agents demand this distribution of
payments for their DR services. We use typical numbers from the PG\&E jurisdiction for residential DR:

\begin{tabular}{lcl}
$\pi^o$  & \$2/agent & recruitment cost \\
$\pi_\text{max}$ & \$1.30 & max DR payment demanded by agents \\
$\Eset [ b ]$ & 5KWh & average reduction/DR event \\
$m$ & 10 & max number of DR events \\
$\pi^e$ & \$0.15 & retail price of electricity \\
$D$ & 100KWh & DR target \\
\end{tabular}

This yields an average cost of DR provision of \$1.51/KWh under baseline-only reporting. This compares unfavorably with
the lower bound $\phi_\text{min}$ = \$0.71/KWh of Theorem \ref{thm:min}.
\hfill $\Box$
\end{example}

The cost of DR provision per KWh under baseline-only reporting can be quite large. It is set by the {\em maximum} marginal utility of consumers in the recruited pool. Doing any better requires agents to reveal their true marginal utilities, allowing the aggregator to set lower reward prices while assuring that agents yield their discretionary consumption. This observation motivates us to consider a more complex mechanism which offers the promise of lower cost DR provision.

\section{Self-Reported Baseline Mechanism (SRBM)} 
\label{sec:mu-unknown}  

\subsection{SRBM Mechanism Definition}
\label{sec:srbm-mech}  

The cost of DR provision under baseline only reporting inspires us to consider a more complex mechanism which we call Self-Reported Baseline Mechanism (SRBM). 
Under SRBM, all agents submit reports $f_k$ of their baselines {\em and} $\mu_k$ of their marginal utilities. 
Agents may not be truthful. 
The key idea is to design the mechanism so that agents {\em reveal their true baselines and marginal utilities}. 
This allows the aggregator to set lower reward prices, without compromising on the delivery of the DR target.
Under SRBM, the reward prices for selected agents is determined by the submitted reports of {\em other agents}. 
In this aspect, SRBM  resembles classic mechanisms such as Vickrey-Clarke-Groves (VCG) auction pricing.

The SRBM mechanism definition is as follows:

{\em Step 1 (Pod Sorting):} \\
Based on the submitted reports, the aggregator sorts agents into $M$ pods, $\Pset^1, \cdots, \Pset^M$. 
A pod is a minimal subset of recruited 
agents that  can deliver the demand reduction target $D$ under SRBM.  We later describe a specific
pod sorting algorithm that delivers the DR target cost effectively in Subsection \ref{subsec:pod}. 

{\em Step 2 (Pod Selection):} \\
The aggregator selects one pod from the set of $M$ pods to deliver the demand reduction target $D$.  Pod $\Pset^i$ will be selected with probability $\beta^i$. The selection probability $\beta^i$ and a specific pod selection mechanism is described in Subsection \ref{subsec:pod}.  

{\em Step 3 (Agent Selection):} \\
Agents in pod $\Pset^i$ are   
{\em sorted in increasing order of their reported marginal utility} $\mu_k$.
Pods are broken into their core and their header, as sown in Figure \ref{fig:pods2}.
The first $k^*$ agents form the pod {\em core} $\Sset^i$ where
\beq
\label{eq:select-m-muuc}
\sum_{k=1}^{k^*} f_k \geq D, \quad \sum_{k=1}^{k^*-1} f_k <  D. 
\eeq
Equivalently, $k^*$ is the smallest index such that the first $k^*$ agents 
from pod $\Pset^i$ in the sorted list of marginal utilities can deliver the target $D$.  If pod $\Pset^i$ is selected in response to a DR event, agents in its core $\Sset^i$ are {\em called on} to provide their demand reduction. The remaining agents form the pod header $\Hset^i$. So, the core and header of the pod are defined as, 
\[ \Sset^i = \{1, \ldots, k^{*} \},~~ \Hset^i = \Pset^i \setminus \Sset^i\]


\begin{figure}[h!]
\centering
\begin{tikzpicture}[scale = 0.8]

\draw [->, very thick] (-2,-0.2) -- (7,-0.2);
\draw [draw=blue, fill=blue, fill opacity = 0.1, thick] (0,0) rectangle (1.9,1);
\draw [draw=red, fill=red, fill opacity = 0.1, thick] (2,0) rectangle (3.9,1);
\draw[thick,<->] (0,1.2) -- (3.9,1.2);
\draw (1.95,1.5) node[align=center] {\scriptsize Pod $\Pset^i$};


%
\draw (1,0.5) node[align=right] {\scriptsize Core $\Sset^i$ };
\draw (3,0.5) node[align=right] {\scriptsize Head $\Hset^i$ };


\draw[thick,->] (3.8,0) -- (3.8,-0.5);
\draw (5.6,-0.75) node[align=center] {\scriptsize $\nu^i$ = maximum reward price};
%
\draw (-0.25,-1) node[align=left] {\scriptsize agents sorted by increasing \\[-0.05in] \scriptsize reported  marginal utility $\mu_k$};
%

\end{tikzpicture}%
\caption{Pod structure: core and header.}
\label{fig:pods2}
\end{figure}
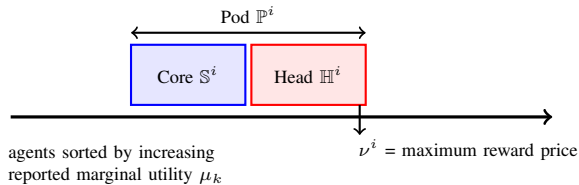

{\em Step 4 (Reward Pricing):} \\
Selected agents (i.e. in $\Sset^i$) receive a reward price $\pi^r_k$ \$/KWh for consumption reduction $(f_k-q_k)^+$ 
below their self-reported baseline. 
Selected agents do not face penalties. 

Let $\Sset_{-k}$ be the set of agents who {\em would have been selected} from pod $\Pset^i$ 
if agent $k$ was not participating in the DR program. Define the reward price for agent $k$ to be
\begin{align}
\label{eq:pslope-main}
& \pi^r_k =  \max\{\mu_j\} -\pi^e \ , \ j \in {\Sset}_{-k}. 
\end{align}
We stress that the reward price $\pi^r_k$ depends on the target $D$, and  is agent-specific.

{\em Step 5 (Penalty Pricing):} \\
Agents who are not selected (i.e. those in the header $\Hset^i$) face a penalty price $\pi^p$ for consumption
deficits $(f_k - q_k)^+$ below their reported baselines.
Under SRBM, the penalty price $\pi^p$ for agent $k$ is chosen to satisfy $\pi^p \geq \pi^e$. 
It is best to select the smallest penalty price, i.e. $\pi^p = \pi^e$, so as not to discourage agents from participating.

This completes the SRBM mechanism definition.

\subsection{Agent Decisions} 
\label{sec:agent-dec}

As with the baseline-only mechanism, agent $k$ solves a two-stage optimization problem. 
In the first stage, it optimally selects its baseline report $f_k$. 
In the second stage, it is informed whether or not it is selected, and then optimally selects its consumption $q_k$ during the DR event. 
The complexity is that reward price $\pi^r_k$ for agent $k$ depends on its submitted reports $(f_k, \mu_k)$ and the reports of all other agents
$(f_{-k}, \mu_{-k})$. These other reports are private information, unavailable to agent $k$.

Suppose agent $k$ submits a baseline report $f_k$. If this agent is selected, its second stage cost function is
\[ J_\text{s}(q_k, f_k) = \pi^e q_k - u(q_k) - \pi^r_k (f_k - q_k)^+. \]
If the agent is not selected, it second stage cost function is
\[ J_\text{ns}(q_k, f_k) = \pi^e q_k - u(q_k) + \pi^p (f_k - q_k)^+. \]
Selected agents are rewarded for consumption reduction from their reported baselines, and agents that are not selected are
penalized for consumption shortfalls below their reported baselines.

Define the optimal consumptions
\begin{eqnarray*}
q^*_\text{s}  =  \arg\min_{q_k} J_\text{s}(q_k, f_k),~~ q^*_\text{ns} =  \arg\min_{q_k} J_\text{ns}(q_k , f_k). 
\end{eqnarray*}
Note that these depend on the first stage decision $f_k$, i.e. the reported baselines.

Let $\beta_k$ be the probability that agent $k$ is selected. Its first stage decision problem is to minimize its expected cost:
\[ J_k(f_k,\mu_k\mid f_{-k}, \mu_{-k} ) =  \beta_k J_\text{s}(q^*_\text{s}, f_k) + (1 - \beta_k) J_\text{ns}(q_\text{ns}^*, f_k)  \]

We use the following notion. 
\begin{definition} {\em [Dominant strategy]} \newline
Let $J_k(f_k,\mu_k\mid f_{-k}, \mu_{-k} )$ be the expected cost for agent $k$ when it reports  $(f_{k}, \mu_{k})$ 
and other agents report $(f_{-k}, \mu_{-k})$. The pair  $(f^{*}_{k}, \mu^*_k)$ is a dominant strategy report for agent $k$ if 
\[ J_k(f_k^*,\mu_k^*\mid f_{-k}, \mu_{-k} ) \leq J_k(f_k,\mu_k\mid f_{-k}, \mu_{-k} ) \]
for all reports $(f_{k}, \mu_{k})$ and $(f_{-k}, \mu_{-k})$.
\end{definition}

The agents have knowledge of only what is informed to them as outlined in section \ref{sec:event-timeline}. In particular each agent is informed the following, { a) that the agent is recruited with a set of agents such that their net reported capacity meets the load reduction requirement, b) that the agents in this group are ordered in increasing order of their marginal utility reports, c) that when this group is called for DR service all or some of the agents in set $\Sset$ which consitute the first $k^*$ agents (as given by \eqref{eq:select-m-muuc}) are selected to provide DR service, d) that the reward and penalty are determined by  \eqref{eq:pslope-main} and $\pi^p_k = \pi^e$, and d) that the probability of calling an agent $k$ in this set $\Sset$ is given by, $\beta_k = \pi^e/(\pi^r_k + \pi^e)$}. \\
Hence under SRBM, { agents are not privy to the pod sorting algorithm}. We call this as the {\it partial information} (PI) setting. Our next result establishes that truthful reporting of baselines and marginal utilities is a dominant strategy for  SRBM under this partial information setting. Mechanism design for truthful reporting under complete information revelation where a consumer is also informed of the pod sorting algorithm is much more complex, and is discussed in the appendix.


We now offer our  main result which establishes key properties of SRBM.
\begin{theorem}\label{thm:muuc-main}
Under the partial information setting (PI), SRBM has the following properties:
\vspace*{-0.08in}
\setlength{\leftmargini}{0.15in}
\balphlist \setlength{\itemsep}{-0.1cm}
\item truthful reporting of baselines and marginal utilities 
is a dominant strategy 
\item called agents consume $q_\text{s}^* = 0$, providing the maximal reduction in their discretionary consumption
\item agents that are not called consume $q_\text{ns}^* = b_k$
\item the aggregator receives no penalty revenue 
\item the load reduction target $D$ is met by each pod core.
\end{list}
\end{theorem}
{\em Proof:} See Appendix \ref{sec:pf:thm:muuc-main}. 

\begin{remark} 
Under SRBM, each consumer reports its true baseline $b_k$. 
As a result, measuring load reduction from the reported baseline coincides with the true load reduction,
and called agents are paid accurately for the services they offer. This makes SRBM fair
Agents who are not called consume $b_k$, which is exactly what they would have if they were not participating in the DR program.
SRBM has a fairness property. An agent who happens to be on vacation will report  their intended consumption, and must incur 
dis-utility to receive a DR payment.
The selection process \eqref{eq:select-m-muuc} meets the load reduction target $D$.
This is because the selected agents completely yield their discretionary electricity consumption during the DR event,  
curtailing their consumption precisely by $b_k$.  
This is an artifact of our piece-wise linear utility function model. 
A more nuanced analysis with general concave utilities will yield a nonzero actual consumption. 
This analysis is more complex and only serves to mask the simplicity of our SRBM.
\end{remark} 

 

\begin{remark} 
We have assumed that agent utility functions (and resulting true baseline consumption $b_{k}$) are deterministic. 
However, $b_{k}$ depends on (exogenous) random parameters such as temperature and occupancy. For example, 
A more realistic model would accommodate dependence on exogenous random processes such as temperature and occupancy.
This might result, for example, in a baseline consumption of the form $b_{k} = \overline{b}_k + a_{k} |\theta - \theta_{0}|$. Here, $\theta$ is the realized temperature during the DR event, and $\theta_{0}$ is the predicted temperature. In this case, agents can be required to
report their best-effort forecast $\overline{b}_k$ of their baseline consumption along with the temperature sensitivity $a_{k}$. 
Historical consumption data can be used to assist agents in making these reports. The SRBM mechanism can  be easily extended to incorporate
these more complex reporting scenarios.  The most general scenarios with uncertain utility functions that
explicitly depend on exogenous random processes $\theta$  is challenging and is an ongoing work.  
\end{remark}  


\subsection{A Pod Sorting Algorithm} \label{subsec:pod} 

We now offer a specific algorithm that sorts recruited agents into pods. From the aggregator's perspective, 
this sorting is attractive as it results in an expected cost of DR provision that is nearly optimal (see Therorem \ref{thm:cc}). The selection probability of a pod $\Pset^i$ is set as, 
\beq \label{eq:pod-sel-prob} \beta^i = \pi^e/\nu^i \leq \beta^i_k  = \pi^e/(\pi^r_k + \pi^e) \ , \ k \in \Sset^i\eeq

\begin{remark}
Pod selection probability $\beta^i$ may not be consistent with the individual selection probabilities of an agent $k$ in $\Sset^i$, which is $\beta^i_k$. When $D$ is large enough and $b$ is small, one can expect $\beta^i_k \sim \beta^i$ and SRBM to be consistent with the individual seleciton probabiliy $\beta^i_k$. In order for SRBM to be consistent, an agent $k (\in \Pset^i)$ is made available for selection, even after $\Pset^i$ is selected for up to $\beta^i$ fraction of time, till the agent $k$'s frequency adds up to $\beta^i_k$. It is in this sense the consumers are informed that {\it all or some consumers} belonging to the pod core would be selected when the pod is selected for DR service (Refer Section \ref{sec:agent-dec}). 
\end{remark}


For a selected pod $\Pset^{i}$ to deliver the DR target, 
the pod core must contain $k^*$ agents, where $\sum_1^{k^*} b_k \geq D$. 
In order to compute the reward prices $\pi^r_k$, the pod header must also contain sufficiently many agents to determine the reward prices $\pi^r_k$.
More precisely, we require that if any agent in the core $\Sset$ is removed, we can still find sufficiently many agents
to determine $\Sset_{-k}$. A sufficient condition for this is that the header $\Hset^i$ contain $k^{**}$ agents, where
$\sum_{k^* +1}^{k^* + k^{**}} b_k \geq D$. 
Thus the selection criterion for pods cores and headers is identical.
{\em  Since agents in the header of a pod are not called, they can serve as the core of another pod.}
This key idea is illustrated in Figure \ref{fig:pods}. It allows us to reduce the total number of agents that must be recruited 
which is one component of the cost of DR provision.

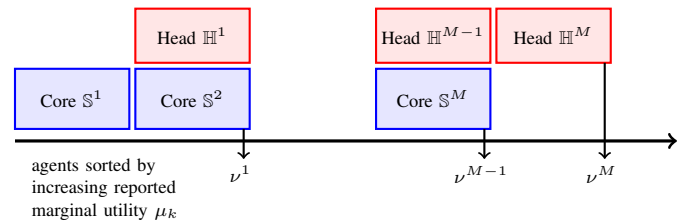
\begin{figure}[h!]
\centering
\begin{tikzpicture}[scale = 0.8]

\draw [->, very thick] (0,-0.2) -- (11,-0.2);
\draw [draw=blue, fill=blue, fill opacity = 0.1, thick] (0,0) rectangle (1.9,1);
\draw [draw=blue, fill=blue, fill opacity = 0.1, thick] (2,0) rectangle (3.9,1);
\draw [draw=blue, fill=blue, fill opacity = 0.1, thick] (6,0) rectangle (7.9,1);

\draw [draw=red, fill=red, fill opacity = 0.1, thick] (2,1.1) rectangle (3.9,2);
\draw [draw=red, fill=red, fill opacity = 0.1, thick] (6,1.1) rectangle (7.9,2);
\draw [draw=red, fill=red, fill opacity = 0.1, thick] (8,1.1) rectangle (9.9,2);
%
\draw (1,0.5) node[align=right] {\scriptsize Core $\Sset^1$ };
\draw (3,0.5) node[align=right] {\scriptsize Core $\Sset^2$ };
\draw (7,0.5) node[align=right] {\scriptsize Core $\Sset^M$ };

\draw (3,1.55) node[align=right] {\scriptsize Head $\Hset^1$ };
\draw (7,1.55) node[align=right] {\scriptsize Head $\Hset^{M-1}$ };
\draw (9,1.55) node[align=right] {\scriptsize Head $\Hset^M$ };

\draw[thick,->] (3.8,0) -- (3.8,-0.5);
\draw (3.8,-0.75) node[align=center] {\scriptsize $\nu^1$ };
\draw[thick,->] (7.8,0) -- (7.8,-0.5);
\draw (7.8,-0.75) node[align=center] {\scriptsize $\nu^{M-1}$ };
\draw[thick,->] (9.8,1.1) -- (9.8,-0.5);
\draw (9.8,-0.75) node[align=center] {\scriptsize $\nu^{M}$ };
%
\draw (1.5,-1) node[align=left] {\scriptsize agents sorted by \\[-0.05in]  \scriptsize increasing reported \\[-0.05in] \scriptsize marginal utility $\mu_k$};
%

\end{tikzpicture}%
\caption{Illustration of pod sorting: cores, headers, and max reward prices.}
\label{fig:pods}
\end{figure}

The second key idea in our pod selection algorithm is to organize agents into pods so that pods with large rewards 
are selected with low probability.
This reduces the expected payout to agents. We begin by sorting agents
in increasing order of their reported marginal utility as shown in Figure \ref{fig:pods}.
The maximum reward price paid to agents in pod $\Pset^i$ is bounded by 
\[ \pi^r_k \leq \nu^i - \pi^e, \ \text{where } \nu^i = \max_{k \in \Pset^i} \pi_k. \]
Also, pod $\Pset^i$ is selected with probability $\beta^i = \pi^e/\nu^i$.
As a result, pods with larger reward prices are selected with lower probability, reducing the expected cost of DR provision.
Agents with high marginal utility are called on less frequently, reducing the expected dis-utility.
This simple sorting algorithm is detailed below in {\it pod sorting algorithm}.

\subsubsection{Pod Selection} A random number $u \sim \mathcal{U}[0,1]$ is drawn. Pod core $\Sset^i$ is selected for DR if $ u \in [ \sum_{j =1}^{i-1} \beta^j, \sum_{j =1}^{i} \beta^j] $. To be consistent with the selection probability of an individual consumer a consumer $k$ in $\Sset^i$ is selected for all draws $ u \in [ \sum_{j =1}^{i-1} \beta^j, \sum_{j =1}^{i-1} \beta^j + \beta^i_k] $. Consumers who are selected are rewarded for reduction from reported baseline based on unit reward announced to them. Consumers who are not selected are penalized for consuming less than the reported baseline. Because the number of pods that are recruited is such that $\sum_i^{M-1} \beta^i < 1 \leq \sum_i^{M} \beta^i $, this pod selection procedure guarantees that the required target $D$ is delivered always.

\renewcommand{\thealgorithm}{}
\begin{algorithm}[!th]
\floatname{algorithm}{}
\caption{\bf Pod Sorting Algorithm}
\label{alg:drmm}
\setlength{\leftmargini}{0.2in}
\begin{list}{\arabic{l1}}{\usecounter{l1} \setlength\labelwidth{0.5in}}
\item Sort agents in  the increasing order of reported marginal utilities
\item Set pod index $i = 1$. Set $n=1$ 
\item Place $k^{*}(i)$ agents indexed from $n$ to $n+k^{*}(i)-1$ in pod core $\Sset^i$  where $k^*(i)$ is the smallest number such that $\dst{ \sum_{j=n}^{n+k^{*}(i)-1} f_{j} \geq D}$. Increment $n \leftarrow n+k^{*}(i)$
\item Place $k^{*}(i+1)$ agents indexed from $n$ to $n+k^{*}(i+1)-1$ in pod header $\Hset^i$  where $k^*(i+1)$ is the smallest number such that $\dst{ \sum_{j=n}^{n+k^{*}(i+1)-1} f_{j} \geq D}$. \\ Increment $n \leftarrow n+k^{*}(i+1)$
\item Define the pod $\Pset^i  = \Sset^i \cup \Hset^i$. Define the pod selection probability $\beta^i = \min \beta^i_k, \ k \in \Sset^i$. Increment $i \leftarrow i+1$
\item Define $\Sset^i  = \Hset^{i-1}$. 
\item If $ \sum_{i} \beta^i < 1$, go to step 4. Else stop.
\end{list} 
\end{algorithm} 

Next, we give a characterization of the number of pods $M$ and number of recruited agents $N$ under this pod sorting algorithm. 

\begin{theorem}
\label{thm:MNbounds}
Let $N$ be the number of recruited agents,  $\psi$ be the payout to agents per DR event, 
and $M$ be the number of pods. Under the pod sorting algorithm,
\begin{align*}
\Eset[N]  &\leq \left(\Eset[M] + 1\right) (D/\Eset [ b ] + 1), \\
\Eset[M] &\leq  \Eset [ \pi ] \ \pi^e + 3. 
\end{align*}
\end{theorem}
\begin{proof}
Refer Appendix
\end{proof}

\subsection{Aggregator Cost} 
%
%
%
%

We are now in a position to compute the expected cost of DR provision under SRBM with our pod sorting algorithm. This given in the following theorem. 

\begin{theorem} \label{thm:cc}
Suppose agent true baselines $\{b_1, b_2, \cdots \}$ and agent true marginal utilities $\{\pi_1, \pi_2, \cdots\}$ 
are i.i.d. random variables. Assume $b_{i}$s and $\pi_{k}$s are independent for all $i, k$. Then, the expected cost $\phi$ per KWh of DR provision 
under SRBM with the pod sorting algorithm satisfies
\beq \label{eq:tmp20} 
\phi_\text{SRBM} \leq \Eset[\pi] + 2 \pi^e  + \left(\frac{\pi^o}{m\Eset [ b ]} + \frac{\pi^o}{mD}\right) \left(\frac{\Eset [ \pi ]}{\pi^e} + 3 \right)
\eeq
\end{theorem}

{\em Proof:} Refer Appendix

\begin{remark}
For large demand reduction targets $D$, our upper bound (\ref{eq:tmp20}) on the average cost of DR provision becomes
\beq \label{eq:803} \phi_\text{SRBM}  \leq \left(\Eset[\pi] +3 \pi^e\right)  - \pi^e +  \frac{\pi^o (\Eset[\pi] +3 \pi^e) }{m \pi^e \Eset[ b ] }  \eeq
If the spread of true marginal utilities is modest, we have argued (see (\ref{eq:802})) that the minimum 
expected cost of DR provision over any mechanism in the class $\Ccal$ is
\beq  \label{eq:803b} \phi_\text{min} \approx \Eset[\pi] - \pi^e + \frac{\pi^o \Eset [\pi ]}{m \pi^e \Eset [b]} \eeq
If the retail electricity price $\pi^e$ is small compared to expected marginal utility $\Eset[\pi]$ of the recruited agents, we have
\[ \phi_\text{SRBM} \approx \phi_\text{min} \]
Comparing expressions (\ref{eq:803}) and (\ref{eq:803b}), 
we conclude that SRBM is {\em nearly optimal} in the metric of expected cost of DR provision.
\end{remark}

\section{Case Study and Simulations}

In this section, we illustrate the  performance advantages of the SRBM mechanism compared to other standard mechanisms through numerical examples and simulations. 


\subsection{SRBM vs Baseline-Only Reporting Mechanism} 

Assume that the marginal utilities of recruited agents are uniformly distributed on
$[0.3,1.3]$ \$/KWh. Equivalently, the mix of recruited agents demand this distribution of
payments for their DR services. We use typical numbers from the PG\&E jurisdiction for residential DR as given in the table below. 

\begin{table}[h]
\label{table:example1}
\centering
\begin{tabular}{|c|c|c|}
\hline
$\pi^o$  & \$2/agent & recruitment cost \\ \hline
$\pi_\text{max}$ & \$1.30 & max DR payment demanded by agents \\ \hline
$\Eset [ b ]$ & 5KWh & average reduction/DR event \\ \hline
$m$ & 10 & max number of DR events \\ \hline
$\pi^e$ & \$0.15 & retail price of electricity \\ \hline
$D$ & 100KWh & DR target \\ \hline
\end{tabular}
\end{table}

This yields an average cost of DR provision of \$1.51/KWh under baseline-only reporting. This compares unfavorably with
the lower bound $\phi_\text{min}$ = \$0.68/KWh of Theorem \ref{thm:min}.

However,  the average cost of DR provision under SRBM in this case is is \$0.84/KWh. 
This is only modestly larger than the {\em lowest possible} average cost for DR provision of  \$0.71/KWh found from Theorem \ref{thm:min}. SRBM compares very favorably with baseline-only reporting 
which has an average cost of \$1.51/KWh. 

\subsection{Numerical Simulations} 

\begin{figure}[h]  
\begin{tabular}{cc}
\includegraphics[scale=0.11]{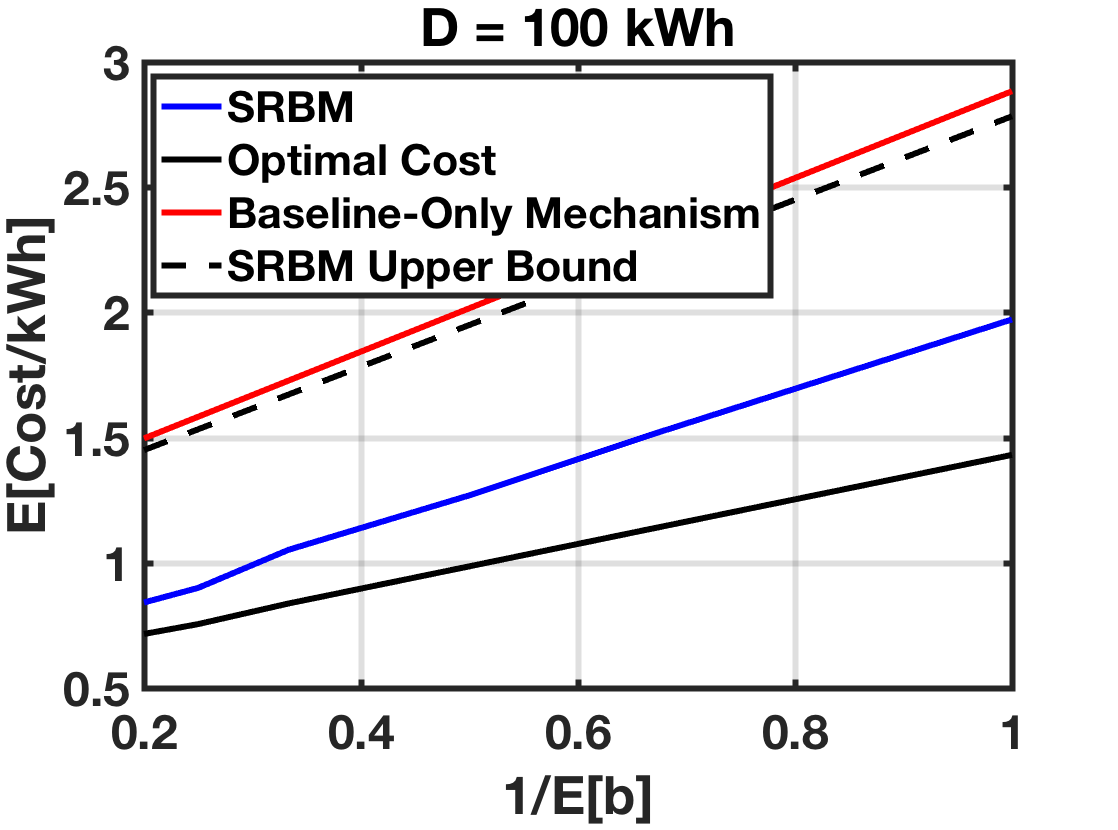} & \includegraphics[scale=0.11]{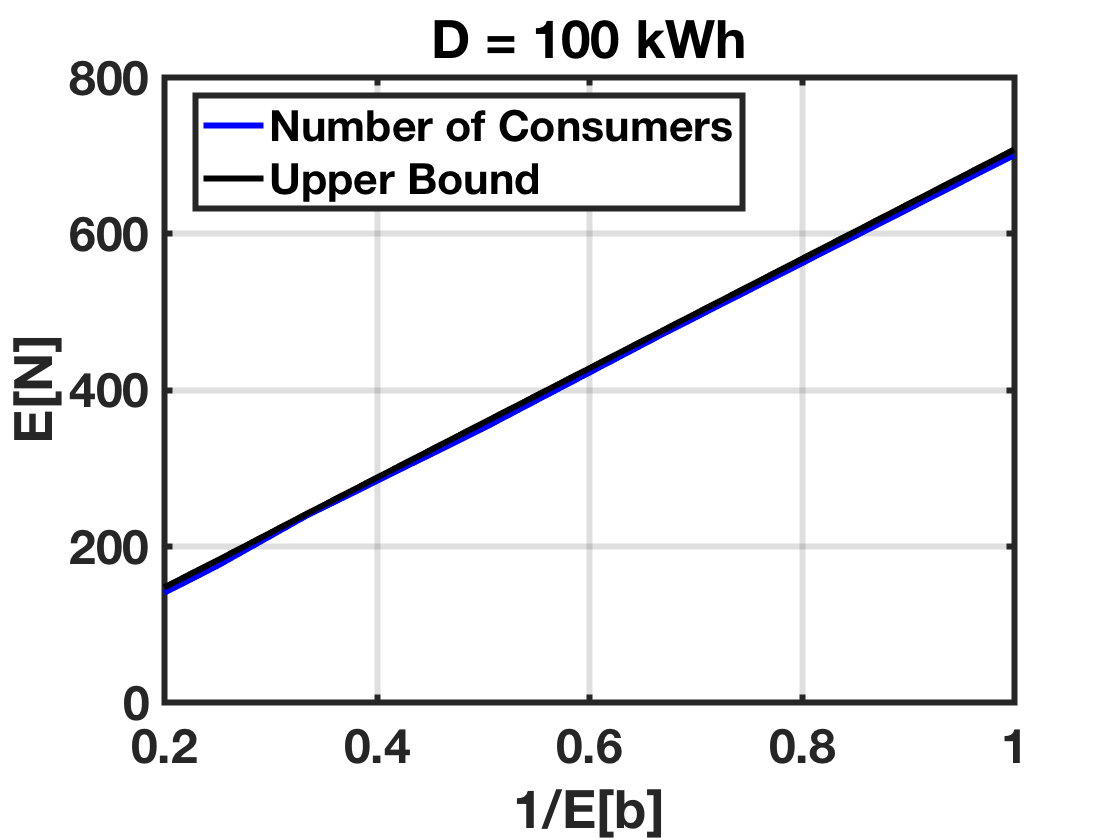} \\
\includegraphics[scale=0.11]{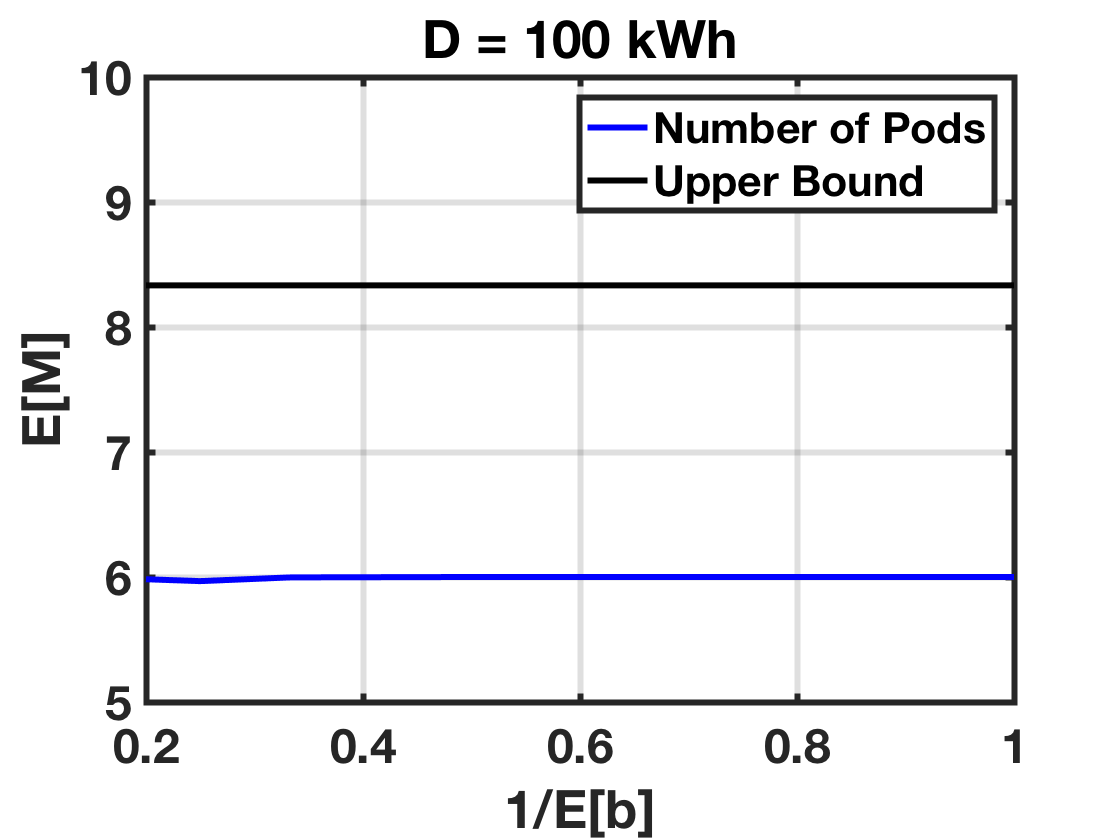} & \includegraphics[scale=0.11]{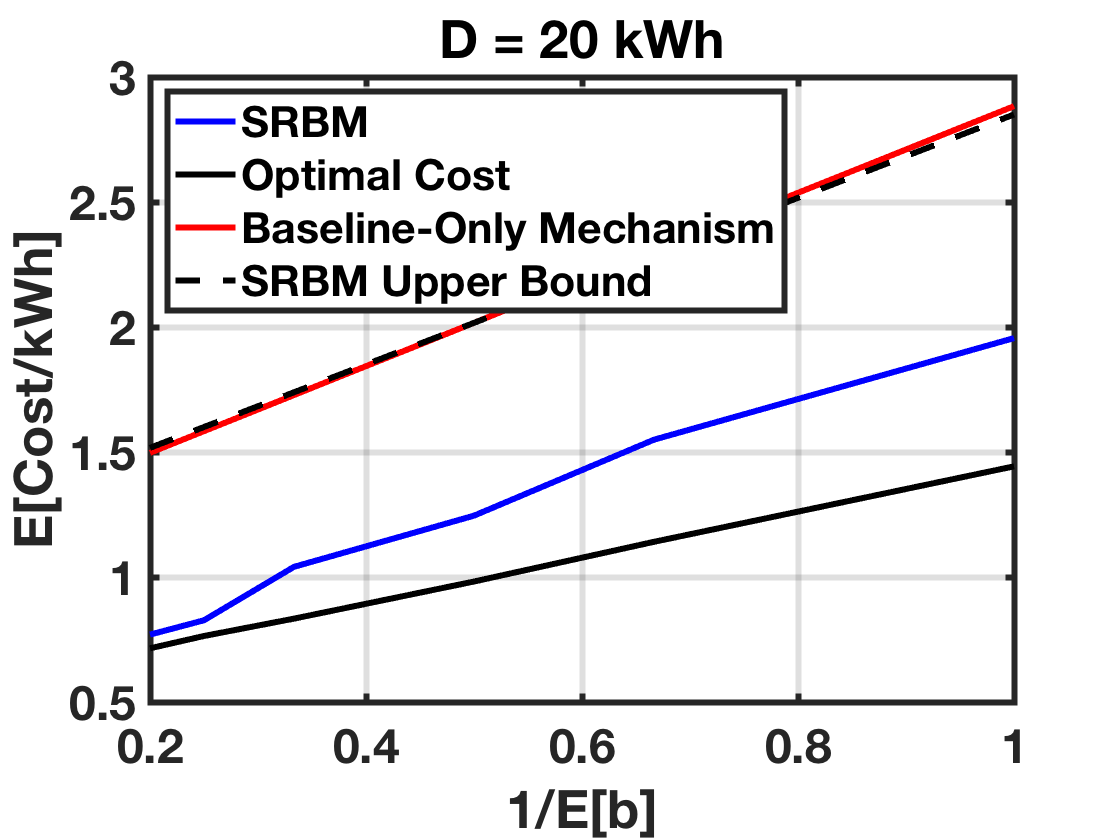} \\
\includegraphics[scale=0.11]{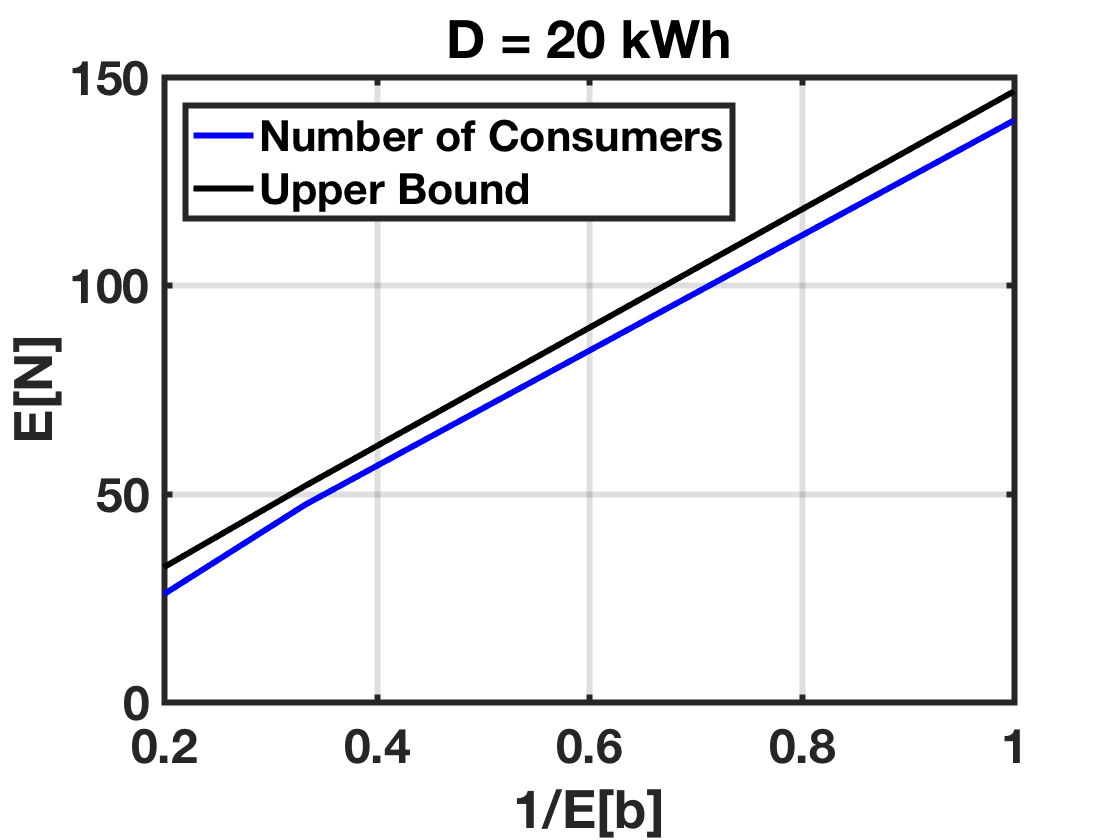} & \includegraphics[scale=0.11]{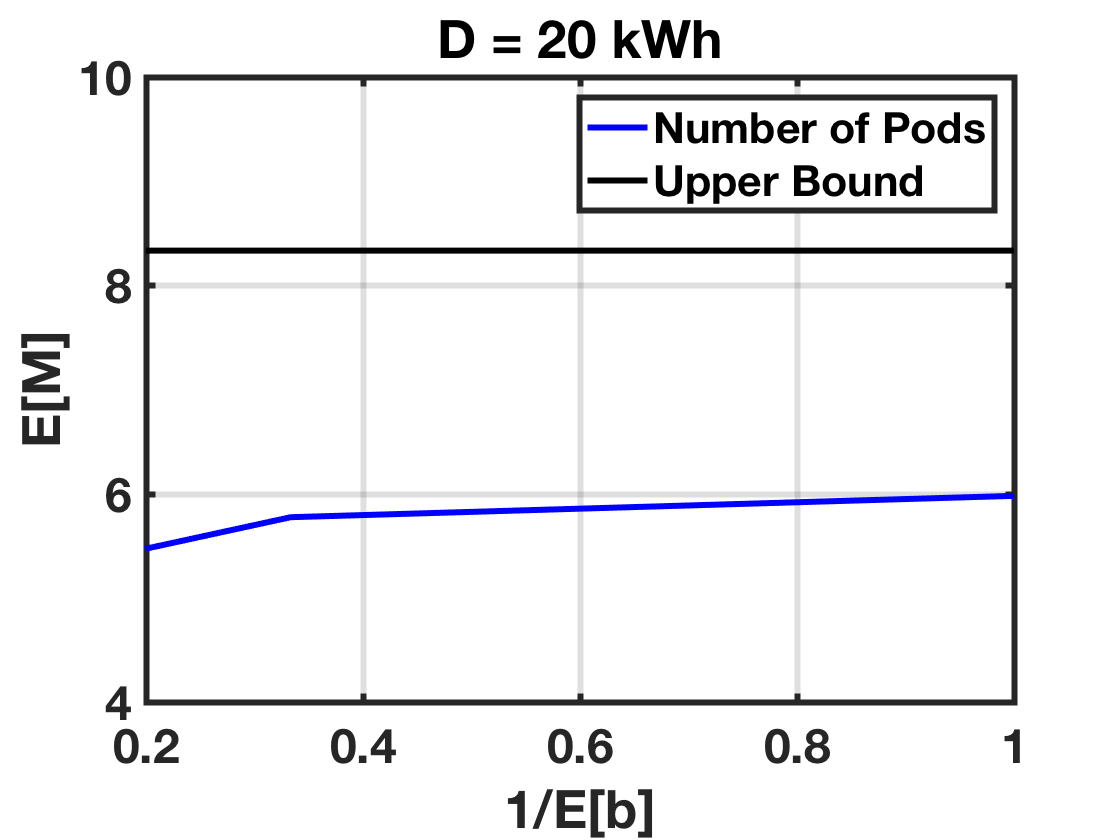} 
\end{tabular}
\caption{Left above: $\Eset$ [Cost/kWh] vs 1/$\Eset$[b] (D = 100 kWh), right above: $\Eset$[N] vs 1/$\Eset$[b] (D = 100 kWh), left middle: $\Eset$[M] vs 1/$\Eset$[b](D = 100 kWh), right middle: $\Eset$[Cost/kWh] vs 1/$\Eset$[b] (D = 20 kWh), left below: $\Eset$[N] vs 1/$\Eset$[b] (D = 20 kWh). right below:  $\Eset$[M] vs 1/$\Eset$[b](D = 20 kWh). Upper bounds are from Theorems \ref{thm:MNbounds} and \ref{thm:cc}}
\label{fig:example-srbm}
\end{figure}

In Fig. \ref{fig:example-srbm}, we provide simulation results that illustrate the superior performance of SRBM mechanism. The first three plots starting from top left provides the average cost, average number of consumers recruited and average number of pods for $D = 100$ kWh (which is representative of a large demand reduction requirement). The next three plots are for $D = 20$ kWh which is representative of smaller demand reduction requirement.  Each plots also contains the theoretical upper/lower bound on the quantities of interests. We use the same values for the prices as in the previous subsection.   

 It follows from the plots that the  cost/kWh of SRBM is significantly lower compared to upper bound and the baseline-only mechanism for all  the scenarios. Moreover, the cost for SRBM is  close to the  optimal cost (given by the lower bound). The figures also show that the  actual number of consumers required for enabling the DR is very close to that predicted by Theorem (\ref{thm:MNbounds}) which further suggests the accurate modeling and analysis of the SRBM mechanism. 
 
Tables \ref{table:CR-SRBM} and \ref{table:CR-SRBM-1} summarizes the result from the simulations. In order to succinctly characterize the  performance optimality, we use the notion of \textit{competitive ratio} (CR) which is the ratio of cost of the mechanism  to the optimal cost possible. The CR is close to 1, though it  increases with $1/\Eset [b]$, which is expected from the upper bound in Eq. \eqref{eq:tmp20}.    

\begin{table}[h] 
\centering
\caption{SRBM Performance:  $D = 20$ kWh }
\begin{tabular}{|c|c|c|c|c|c|}
\hline
$\Eset$[b] & 5 & 4 & 3 & 2 & 1 \\
\hline
 SRBM & 0.77 & 0.83 & 1.04 & 1.24 & 1.96 \\
 \hline
 $\phi_{min}$ & 0.71 & 0.76 & 0.84 & 0.99 & 1.44 \\
 \hline
CR & 1.08 & 1.09 & 1.24 & 1.26 & 1.36 \\
\hline
\end{tabular}
\label{table:CR-SRBM-1}
\end{table}

\begin{table}[!h]
\centering
\caption{SRBM Performance:  $D = 100$ kWh }
\begin{tabular}{|c|c|c|c|c|c|}
\hline
$\Eset$[b] & 5 & 4 & 3 & 2 & 1  \\
\hline
 SRBM & 0.84 & 0.9 & 1.05 & 1.27 & 1.97 \\
 \hline
 $\phi_{min}$ & 0.72 & 0.76 & 0.84 & 0.99 & 1.43 \\
 \hline
CR & 1.18 & 1.19 & 1.26 & 1.29 & 1.38 \\
\hline
\end{tabular}
\label{table:CR-SRBM}
\end{table}

\subsection{Self-Reported vs CAISO Baseline} 

California Independent System Operator (CAISO) estimates baseline by the $10/ 10$ method. In the $10/ 10$ method, a raw baseline $\bar{b}_c$ is estimated by averaging the consumption of $10$ non-event similar days closer to the event day. In addition an adjustment factor is calculated to account for any variation on the event day. This adjustment factor is calculated based on the consumption $q_p$ couple of hours prior to the event and the average consumption $\bar{b}_p$ at the same time on the selected non-event days. This factor is capped at $20\%$ both upwards and downwards. With this baseline estimation approach, the payment $R_c$ to a consumer is given by $R_c = \pi^r_k \left(\frac{q_p}{\bar{b}_p}\bar{b}_c - q\right)$. Lets consider the same example as before, where $\pi \sim \mathcal{U}[0.3,1.3]$ and $\pi^e = 0.15$. Lets also consider the case where the uninflated value of $q_p = \bar{b}_c = \bar{b}_p$ and that $\bar{b}_c$ and $\bar{b}_p$ are the true values. Here we ignore any incentives from the consumer side to alter $\bar{b}_c$ and $\bar{b}_p$. From Eq. \eqref{eq:pslope-main} it is clear that $\pi^r_k \geq \pi -\pi^e \geq \pi^e$ for each consumer. One the day of the event, if an individual consumer inflates its consumption $q_p$ then it gains at the rate of $\pi^r_k$ for each unit of increase and its loss is at the rate of retail price $\pi^e$ (follows from the piecewise linear utility model Eq. \eqref{eq:utility-1}). Here the gain and the loss is evaluated in terms of the net utility for the hours before the DR event and during the DR event. Since $\pi^r_k \geq \pi^e$ increasing $q_p$ results in a net positive gain for the consumer on the event day. Then each consumer will increase $q_p$ by the maximum allowed that is $20 \%$ resulting in $20 \%$ inflation of baseline and payments, while self-reporting nullifies any such incentive to inflate (Theorem \ref{thm:muuc-main}). 

\section{Conclusion}
\label{sec:conclusions}
In this paper, we have addressed the baseline estimation problem that is central to demand response programs. 
We proposed a mechanism where agents participating in a DR program self-report their baselines and  marginal utilities. 
Under this self reported baseline mechanism (SRBM), agents reveal their true baselines and marginal utilities. Agents that are selected for demand reduction maximally curtail their load. When agents are not selected, their actual consumption is  faithful to their baselines. 
As a result, the aggregator is able to meet any feasible load reduction target with certainty by selecting agents whose measured reductions adds up to the total reduction. We have also derived a lower bound $\phi_\text{min}$ on the expected cost per KWh of DR provision 
for any mechanisms in the class $\Ccal$ with piece-wise linear penalty and reward functions. We argue that SRBM comes close to deliver 
DR provision at this minimum cost.

Our work has assumed that the true baseline consumption is deterministic and that the utility function is piece-wise linear. 
In practice, DR mechanisms must account for uncertainty in agents intended consumption during the DR delivery window.
The exploration of the DR mechanisms for general utility functions $u(q;\theta)$ where $\theta$ is an exogenous random variable remains to be explored. 



\bibliographystyle{IEEEtran}
\bibliography{Refs-BL-TSG}

\appendix
\subsection{Mathematical Preliminaries}

We require two preliminary results based on the classical optional stopping theorem and the closely related Wald's equation. 

\begin{proposition} (\em Wald's equation) \label{thm:wald}
Let $X = \{ X_k:  k = 1,2, \cdots \}$ be a sequence of independent identically distributed random variables. Let $\chi = \Eset [ X_k]$.
Let $N$ be a stopping time with respect to $X$, i.e. an integer valued random variable that depends only on $\{X_1, \cdots, X_N\}$.
Suppose $\Eset [ N ] < \infty$. Then
\beq \label{eq:wald} \Eset [X_1 + \cdots +X_N] =  \chi \Eset[N] \eeq
\end{proposition}

{\em Proof:} Theorem (1.6), Chapter 3 in \cite{durrett2010probability} \hfill $\Box$

\begin{proposition} \label{thm:optional}
Let  $X = \{ X_k:  k = 1,2, \cdots \}$ be a sequence of independent identically distributed random variables.
Let $\chi = \Eset [ X_k ] $ and define the partial sums 
\[ S_0 = 0, S_1 = X_1, S_2 = X_1 + X_2, \cdots \]
Fix $D$ and consider the random variable 
\[ N = \min\{t: S_t \geq D \}. \]
Then,  
\[ \frac{D}{\chi} \leq \Eset [ N ] < \frac{D}{\chi} + 1 \]
\end{proposition}

{\em Proof:} We prove the result in three steps.

a) {\em First we show that $\Eset [N]$ is bounded.} \\
Construct the sequence of i.i.d random variables defined by
\begin{equation}
Y_k = \left\{ \begin{array}{cc} 0 & \ \text{If} \ X_k < a \\  a & \ \text{If} \ X_k \geq a \end{array}\right.
\end{equation} 
Call a trial $k$ a success if $Y_k = a$. 
Let $T$ denote the first success in a sequence of trials. Denote the probability of success on a given trial by $p$.
Then $p = \Pset(Y_k = a)$ and $T$ is a geometric random variable with probability distribution
\begin{equation}
\Pset(T = k) = (1 - p)^{k-1}p.
\label{eq:dist-tr2suc} 
\end{equation}
From \eqref{eq:dist-tr2suc} it follows that 
\[ \Eset [T] = 1/p = 1/\Pset(X_k \geq a). \]

Define, $R_{t} = \sum_{k = 1}^{t} Y_k$ and $M = \min\{t: R_{t} > D\}$. Then 
\begin{equation}
\label{eq:dk671}
M \leq \text{Number of trials needed for} \ \ceil{D/a} + 1 \ \text{successes}
\end{equation}
Denote $T_{k-1}^{k} $ as the number of trials to reach $k$ th success after the $(k-1)^\text{th}$ success. Then from 
\eqref{eq:dk671}
%
\begin{equation}
\Eset [M] \leq \Eset \sum_{k = 1}^{\ceil{D/a} + 1} T_{k-1}^{k}  = \sum_{k = 1}^{\ceil{D/a} + 1} \Eset [T_{k-1}^{k}]
\label{eq:ub-tildeN-1}
\end{equation}
Observe that $T_{k-1}^{k} $ and $T$ are identically distributed. As a result, $\Eset[T_{k-1}^{k}]  = \Eset[T]$ and
\begin{equation}
\Eset [M] \leq \sum_{k = 1}^{\ceil{D/a} + 1} \Eset [T] = \left(\ceil{D/a} + 1\right)/ \Pset(X_k \geq a)
\end{equation}
From the construction of $Y_n$ it trivially follows that $Y_n \leq X_n$ which implies that $R_t \leq S_t$. So clearly, $N \leq M$. So,
\begin{equation}
\Eset [N] \leq \Eset [M] \leq \left(\ceil{D/a} + 1\right)/ \Pset(X_k \geq a) < \infty
\end{equation}

b) {\em It is straightforward to show that  $N$ is a stopping time with respect to the sigma algebra generated by $X_k, k = \{1,2, . . . ,n\}$}. We omit the proof. 


c) {\em We show that $\frac{D}{\chi} < \Eset [ N ] \leq \frac{D}{\chi} + 1 $}:

Since $N$ is a stopping time, 
\[ \Eset[ S_N]  = \chi \Eset[N]  \geq D \implies \Eset[N] \geq D/\chi \]
Since $(N-1)$ is also a stopping time, 
\[  \Eset[S_{N-1}] =  (N-1) \chi < D  \implies \Eset]N] <  D/\chi + 1. \]
\hfill $\Box$

%
\subsection{Three Intermediate Results} 


We first explore second-stage decision of agents {\em under truthful reporting of baselines}.
We show that if the reward price exceeds the marginal utility, selected agents completely yield their discretionary consumption.
Conversely, if $\pi^r < \pi_k - \pi^e$, the reward price is insufficient to persuade agents to offer any demand reduction.
If the penalty prices is non-negative,  participation in the DR program does not alter the energy consumption of 
agents who are not selected.
We have the following:

\begin{proposition} 
\label{prop:2} Suppose agent $k$ faces a penalty price $\pi^p$ and receives a reward price $\pi^r$.
Assume agent $k$ reports its baseline truthfully, i.e. $f_k = b_k$.
\setlength{\leftmargini}{0.11in}
\balphlist \setlength{\itemsep}{-0.1cm}
\item Suppose agent $k$ is not selected. If $\pi^p \geq 0$, its unique optimal second-stage decision and resulting optimal cost are
\[ q^*_\text{ns} = b_k, \ J_\text{ns}(q^*_\text{ns}) = -(\pi_k - \pi^e) b_k \]
\item Suppose agent $k$ is selected. If $\pi^r > \pi_k - \pi^e$, its unique optimal second-stage decision and resulting optimal cost are
\[ q^*_\text{s} = 0, \ J_\text{s}(q^*_\text{s}) = -\pi^r  b_k \]
\item Suppose agent $k$ is selected. If $\pi^r < \pi_k - \pi^e$, agent $k$ will not yield any of 
its discretionary consumption, i.e. $q^*_\text{s} = b_k$.
\end{list}
\end{proposition}

{\em Proof:} (a) If agent $k$ is not selected, its cost function is
\[ J_\text{ns}(q) = \pi^e q - u(q) + \pi^p(b_k - q)^+ \]
which is piece-wise linear in the decision variable $q$. We compute the slopes:
\begin{eqnarray*}
 \frac{d J_\text{ns}(q)}{d q} & = & \left\{ \begin{array}{cl} 
 			 \pi^e - \pi_k  - \pi^p   & \text{if } q < b_k \\ 
 			\pi^e   & \text{if } q > b_k \end{array} \right.  \\
 			& =& \left\{ \begin{array}{cl} 
 			< 0   & \text{if } q < b_k \\ 
 			> 0  & \text{if } q > b_k \end{array} \right.
\end{eqnarray*} 
Thus $q_\text{ns}^* = b_k$ is the unique minimizer of the cost function $J_\text{ns}$. The resulting optimal cost is clearly
$-(\pi_k - \pi^e) b_k$. \\[0.1in]
(b) If agent $k$ is selected, its cost function is
\[ J_\text{s}(q) = \pi^e q - u(q) - \pi^r(b_k - q)^+ \]
This is piece-wise linear in the decision variable $q$. We compute the slopes:
\begin{eqnarray*}
 \frac{d J_\text{s}(q)}{d q} & = & \left\{ \begin{array}{cl} 
 			 \pi^e - \pi_k  + \pi^r   & \text{if } q < b_k \\ 
 			\pi^e   & \text{if } q > b_k \end{array} \right.  \\
 			& =& \left\{ \begin{array}{cl} 
 			> 0   & \text{if } q < b_k \\ 
 			> 0  & \text{if } q > b_k \end{array} \right.
\end{eqnarray*} 
Thus $q_\text{s}^* = 0$ is the unique minimizer of the cost function $J_\text{s}$. The resulting optimal cost is clearly
$-\pi^r  b_k$.  \\[0.1in]
(c) If $\pi^r < \pi_k - \pi^e$, we have
\begin{eqnarray*}
 \frac{d J_\text{s}(q)}{d q} & = & \left\{ \begin{array}{cl} 
 			< 0   & \text{if } q < b_k \\ 
 			> 0  & \text{if } q > b_k \end{array} \right.
\end{eqnarray*} 
Thus $q_\text{s}^* = b_k$ is the unique minimizer of  $J_\text{s}$.
 \hfill $\Box$

We next explore conditions on the penalty and reward prices, and the selection probability which {\em induce agents to 
truthfully report their baselines}.

\begin{proposition}  
\label{prop:1}
Assume the mechanism prices and selection probability satisfy 
$ \pi^p \geq \pi^e,   \alpha \leq (\pi^e)/(\pi^r + \pi^e)$.  Then,
\setlength{\leftmargini}{0.13in}
\balphlist \setlength{\itemsep}{-0.1cm}
\item an optimal first stage decision  for  agent $k$ is to truthfully report its baseline, i.e. $f_k^* = b_k$. 
\item if further, $\pi^r > \pi_k - \pi^e$, truthful reporting of baselines is the unique optimal decision.
\end{list}
\end{proposition}

{\em Proof:} (a)
Consider agent $k$, and fix the submitted first-stage baseline report $f$. Define the constant
\[ c = \pi^e b_k - u(b_k) = (\pi^e - \pi_k) b_k. \] 
The selection probability inequality can be re-written as
\begin{equation} \label{eq:tmp1}
-\alpha \pi^r + (1 - \alpha) \pi^e  \geq 0.
\end{equation}
We first compute the optimal consumption decided by agents in the second stage in two cases:
\vspace{-0.4cm}
\setlength{\leftmargini}{0.13in}
\bromalist
\item {\em Agent is not called.} Its cost function is
\[ J_\text{ns}(q,f) = \pi^e q - u(q) + \pi^p ( f -q)^+  \]
Since this is a continuous piece-wise linear function, it easy to verify that the optimal consumption is either $q^* = b_k$ or $q^* = f$.
The optimal cost is:
\begin{eqnarray*} 
& & \hspace*{-0.45in} J_\text{ns}(q^*,f) = \min \{ c + \pi^p (f-b_k)^+,  \pi^e f - u(f) \} \\
		& = & \left\{ \begin{array}{cl} \min \{ c + \pi^p (f-b_k),  c + \pi^e (f-b_k) \} & \text{if } f >  b_k \\  \min \{ c ,  c + (\pi_{k} - \pi^e)(b_k-f)  \}  & \text{if } f< b_k \end{array} \right. \\
		& = & \left\{ \begin{array}{cl} c+ \min \{ \pi^p,  \pi^e \} \cdot (f - b_k) & \text{if } f >  b_k \\  c+ \min \{0, (\pi_{k} - \pi^e)(b_k-f) \}  & \text{if } f< b_k
		\end{array} \right. \\
		& = & \left\{ \begin{array}{cl} c+ \pi^e  (f - b_k) & \text{if } f > b_k \\  c & \text{if } f< b_k
		\end{array} \right.
\end{eqnarray*}
\item {\em Agent is called.}  Its cost function is
\[ J_\text{s}(q,f)  = \pi^e q - u(q) - \pi^r (f - q)^+ \]
Since this is a continuous piece-wise linear function, it easy to verify that the optimal consumption is either $q^* = 0$ or $q^* = b_k$. 
The optimal cost is:
\begin{eqnarray*}
J_\text{s}(q^*; f) & = & \min \{-\pi^r f,  c - \pi^r (f - b_k)^+ \}  \\ 
		& = & \left\{ \begin{array}{cl} \min \{-\pi^r f,  c - \pi^r (f - b_k) \}  & \text{if } f > b_k \\ \min \{-\pi^r f,  c   \} & \text{if } f< b_k \end{array} \right. \\ 
		& = &  \left\{ \begin{array}{cl} -\pi^r f + \min \{0,  c + \pi^r b_k \}  & \text{if } f > b_k \\ \min \{-\pi^r f,  c  \} & \text{if } f< b_k \end{array} \right. 
\end{eqnarray*}
\end{list}
Combining (i) and (ii), we can write the expected cost as
\begin{eqnarray*}
 J(f) &=& \alpha J_\text{s} + (1-\alpha) J_\text{ns} \\
 & = & \left\{ \begin{array}{cl} 
 			 -\alpha \pi^r f + \alpha \min \{0,  c + \pi^r b_k \} & \\
 			 \quad + (1 - \alpha) (c+  \pi^e  (f - b_k))   & \text{if } f > b_k \\ 
 			 \alpha \min \{-\pi^r f,  c  \} + (1 - \alpha) c & \text{if } f< b_k \end{array} \right. 
\end{eqnarray*} 
Note that $J(f)$ is continuous and piece-wise linear. We compute the slopes
\begin{eqnarray} \label{eq:tmp99}
 \frac{d J(f)}{df} & = & \left\{ \begin{array}{cl} 
 			 -\alpha \pi^r + (1 - \alpha) \pi^e   & \text{if } f > b_k \\ 
 			\text{either } -\alpha \pi^r \ \text{or } 0  & \text{if } f< b_k \end{array} \right.  \\
 			& =& \left\{ \begin{array}{cl} 
 			\geq 0   & \text{if } f > b_k \\ 
 			\leq 0  & \text{if } f< b_k \end{array} \right. \nonumber
\end{eqnarray} 
Here, we have used (\ref{eq:tmp1}). Consequently, $f^* =b_k$ is a minimizer of $J(f)$. 

(b) Next suppose $\pi^r > \pi_k - \pi^e$. Notice that
\[ c + \pi^r b_k = (\pi^e - \pi_k) b_k + \pi^r b_k > 0.\]
Define $ g = -c / \pi^r$, and notice that for $f<g$,
\[ c + \pi^r f < c + \pi^r g = 0. \]
It is easy to verify that the expected cost simplifies to
\begin{eqnarray*}
 J(f) \hspace*{-0.15in}
 & = & \hspace*{-0.15in}  \left\{ \begin{array}{cl} 
 			 -\alpha \pi^r f + (1 - \alpha) (c+  \pi^e  (f - b_k))   & \text{if } f > b_k \\ 
 			 - \alpha \pi^r f + (1 - \alpha) c & \text{if } g< f< b_k \\
 			 c & \text{if } f < g \end{array} \right. 
\end{eqnarray*} 
Note that $J(f)$ is continuous and piece-wise linear. We compute the slopes
\begin{eqnarray*}
 \frac{d J(f)}{df} & = & \left\{ \begin{array}{cl} 
 			 -\alpha \pi^r + (1 - \alpha) \pi^e   & \text{if } f > b_k \\ 
 			 -\alpha \pi^r  & \text{if } g< f< b_k \\
 			 0 & \text{if }  f< g  \end{array} \right.  
\end{eqnarray*} 
Consequently, $f^* =b_k$ is the unique minimizer of $J(f)$.
\hfill $\Box$

Finally, we show that if the selection probability is too large, agents will arbitrarily inflate their baselines.
\begin{proposition} \label{prop:3}
Assume the selection probability satisfies \\ $ \alpha > (\pi^e)/(\pi^r + \pi^e)$. Then the optimal first stage decision is to report the maximum, i.e., $f^{*}=b_{\max}$. 
\end{proposition}
{\em Proof:} 
We follow the calculation of Proposition \ref{prop:1}. The selection probability inequality can be re-written as
\[  -\alpha \pi^r + (1 - \alpha) \pi^e < 0. \]
It follows from (\ref{eq:tmp99}) that 
\begin{eqnarray*} 
 \frac{d J(f)}{df} & = & \left\{ \begin{array}{cl} 
 			< 0   & \text{if } f > b_k \\ 
 			\leq 0  & \text{if } f< b_k \end{array} \right. \nonumber
\end{eqnarray*} 

Thus $J(f)$ is monotone decreasing. The optimal first-stage decision is $f^* = b_{\max}$, i.e. agents will arbitrarily inflate their baselines. \hfill $\Box$

\subsection{Proof of Theorem \ref{thm:1}} 
\label{sec:pfthm:1}

{\em Proof:} (a) From Proposition \ref{prop:1}, we know that  an optimal first stage decision is for agents to truthfully 
report their baselines, i.e. $f^*_k = b_k$. \\
(b) and (c) now follow immediately from Proposition \ref{prop:2}. \\
(d) is immediate on observing that $\pi^p(f^*_k - q^*_\text{ns}) = 0$. \\
(e) Recall that agents are selected at random until (see (\ref{eq:aa}))
\[ \sum f_k \geq D. \]
Note that (a) agents report baselines truthfully, and (c) selected agents completely yield their discretionary consumption.
As a result, the set of selected agents provide $\sum_k b_k$ in demand reduction, meeting the target $D$. \hfill $\Box$

\label{sec:thm:min}

{\em Proof:} We consider mechanisms where agents that are called receive a payment $\pi_k^r(b_k - q_k)^+$ 
for consumption reduction from their true baseline. The reward price can be agent specific.

From Proposition \ref{prop:2} (c), if the reward price is too low, i.e. $\pi^r_k< \pi_k - \pi^e$, agent $k$ will not yield any of its discretionary consumption to the DR program.
So, we must have $\pi_k^r \geq \pi_k - \pi^e$, and under this condition, selected agents completely yield their discretionary consumption.

If agent $k$ knows with {\em certainty} that it will be called for a DR event, its first-stage cost function is
\[ J(f_k, q_k) = - \pi_r^k (f_k - q_k)^+ + \pi^e q_k - u(q_k). \]
As its optimal second stage decision is $q^* = 0$, the resulting first-stage decision problem is
\[ \min_{f_k} -\pi_r^k f_k. \]
This implies that agent $k$ will arbitrarily inflate its baseline, resulting in an unbounded expected cost for DR provision. 
As a result, we must allow for agent $k$ to be selected {\em at random} with a sufficiently small
probability. From Proposition \ref{prop:3}, the selection probability $\alpha_k$ for agent $k$ must satisfy 
$\alpha_k  \leq (\pi^e)/(\pi^r_k + \pi^e)$, else agent $k$ will arbitrarily inflate its baseline, resulting in an unbounded expected cost for DR provision.

To minimize its DR payments to agents, 
the aggregator should select the smallest possible reward, i.e. it should choose $\pi^r_k = \pi_k - \pi^e$. 
To minimize its recruitment costs, the aggregator should recruit the smallest number of agents or maximize the selection probability. Thus we set  
\[ \alpha_k  = \frac{\pi^e}{\pi^r_k + \pi^e} = \frac{\pi^e}{\pi_k} \]
Any other mechanism in $\Ccal$ would result in a higher expected cost for DR provision.

With these choices, under any mechanism in the class $\Ccal$, the expected 
demand reduction from agent $k$ is $\alpha_k b_k$, and the
expected payout to agent $k$ is $\alpha_k \pi_k^r b_k$. The expectations here are over 
the randomness of being selected under the DR mechanism.

Define the random variables $X_k = b_k w_k$ where $w_k = \pi^e/\pi_k$. Notice that $ \{b_1, b_2, \cdots\}$ and
$\{w_1, w_2, \cdots \}$ are  i.i.d. Also, $b_{i}$s and $w_{k}$s are independent.  Thus $\{X_1, X_2, \cdots \}$ is i.i.d. and
\beq \label{eq:tmp89}  \Eset [ X] :=  \Eset [ X_{k} ] = \pi^e \Eset[ b] \Eset [1 /\pi ] \eeq

To meet the load reduction target $D$, we must recruit $N$ customers where
\[ \sum_1^N \alpha_k b_k = \sum_1^N \frac{\pi^e}{\pi_k} b_k  = \sum_1^N X_k \geq D. \]
Using Proposition \ref{thm:optional}, we have
\beq \label{eq:tmp55} \Eset [ N ] \geq \frac{D}{\Eset[X_{k} ]} = \frac{D}{\pi^e \Eset [ b ] \Eset [ 1 /\pi ] }. \eeq
The expected payout $\psi$ to the agents (over the random variables $b_k, \pi_k$) is
\begin{eqnarray*}
\Eset[ \psi]  & = & \Eset \left[ \sum_1^N \alpha_k \pi_k^r b_k \right]  =    \Eset \left[ \sum_1^N \frac{\pi^e (\pi_k - \pi_e) }{\pi_k}  b _k \right] \\
 & =  &  \Eset \left[ \sum_1^N (\pi^e b_k  - \pi^e X_k) \right] \\
 & = &  \pi^e \Eset [ N ]  \left( \Eset [ b ] -   \Eset [X ] \right)  \\[0.05in]
 & = & \pi^e \Eset [ N ] \Eset [ b ] \left( 1 -   \pi^e \Eset[1/\pi] \right)
\end{eqnarray*}
In the final step above, we have used Wald's equation (see Proposition \ref{thm:wald}) together with 
(\ref{eq:tmp89}). Combining this with (\ref{eq:tmp55}), we obtain
\[  \Eset[ \psi]  \geq \frac{D}{\Eset [ 1 /\pi ]}  - \pi^e D.  \]

%

We can now compute a lower bound on the expected cost of DR provision:
\begin{eqnarray*}
\phi & = & \underbrace{\frac{\Eset [ \psi ]}{D}}_\text{payout per KWh} + \underbrace{\frac{\pi^o\Eset[N]}{mD}}_\text{recruitment cost}\\
& \geq &  \frac{1}{\Eset[ 1/\pi] } - \pi^e + \frac{\pi^o}{m \Eset [b]}\cdot \frac{1}{\pi^e\Eset [ 1/\pi]}
\end{eqnarray*}
proving the claim. \hfill $\Box$

\subsection{Proof of Theorem \ref{thm:aa}} 
\label{sec:pf:thm:aa}
{\em Proof:}   To minimize its DR payments to agents, 
the aggregator should select the smallest possible reward, i.e. it should choose (see Theorem \ref{thm:1}) $\pi^r = \pi_\text{max} - \pi^e$. To minimize its recruitment costs, the aggregator should recruit the smallest number of agents or maximize the selection probability. 
Thus we set (see Theorem \ref{thm:1}), $ \alpha  = (\pi^e)/(\pi^r + \pi^e) = (\pi^e)/(\pi_\text{max})$. 
The number of recruited agents needed to service the DR target is the smallest integer $N$ such that 
\[ \sum_1^N \alpha b_k \geq D.  \]
From Proposition \ref{thm:optional}, we have
\[ \Eset [ N ] <  \frac{D}{\alpha \Eset [b]} +1  = \frac{D \pi_\text{max}}{\pi^e \Eset[ b ] } + 1. \]
If agent $k$ is selected, it delivers $b_k$ KWh of demand reduction. The expected payout to the agents is then
From Theorem \ref{thm:wald},  the expected payout is
\[ \Eset [ \psi ] = \Eset [ \sum_1^N \pi^r \alpha b_k ] = \pi^r \Eset [N] \alpha \Eset[ b ]  <    \pi^r D +  \frac{\pi^r \pi^e \Eset[ b ] }{\pi_\text{max}}    \]

We can now compute an upper bound on the expected cost of DR provision:
\begin{eqnarray*}
\phi & = & \underbrace{\frac{\Eset [ \psi ]}{D}}_\text{payout per KWh} + \underbrace{\frac{\pi^o\Eset[N]}{mD}}_\text{recruitment cost}\\
&\leq& (\pi_\text{max} - \pi^e)  \left(1 + \frac{ \pi^e \Eset[ b ] }{\pi_\text{max}D } \right)   + 
\frac{\pi^o \pi_\text{max} }{m \pi^e \Eset[ b ] }  + \frac{\pi^o}{mD}.
\end{eqnarray*}

\subsection{Proof of Theorem \ref{thm:muuc-main}} 
\label{sec:pf:thm:muuc-main}

{\em Proof:} Under SRBM, we have $\pi^p \geq \pi^e$. We first bound the selection probability. Suppose agent $k$ in pod $\Pset^i$ is selected. 
The reward price offered to agent $k$ satisfies
\[ \pi_k^r \leq  \max_{k \in \Pset^i} \mu_k  - \pi^e = \nu^i - \pi^e. \]
The probability $\beta_k$ that agent $k$ is called satisfies
\[ \beta_k = (\pi^e)/(\pi_k^r + \pi^e). \]
This is also independent of agent $k$'s reports. We therefore meet all the conditions of Proposition \ref{prop:1}. We conclude that
all agents will report its baselines truthfully, or $f_k^* = b_k$ for all $k$. More precisely, we have
\[ J(\mu_k, f_k \mid \mu_{-k}, f_{-k}) \geq J( \mu_k, b_k \mid \mu_{-k}, f_{-k}). \]
We next show that agents are also best served by reporting their
marginal utility truthfully, i.e. 
\[ J( \mu_k, b_k  \mid \mu_{-k}, f_{-k}) \geq J(\pi_k, b_k  \mid \mu_{-k}, f_{-k}). \]

Observe that the reward price for agent $k$ under SRBM is determined by the reports of other agents, $i \neq k$.
Suppose agent $k$ is selected. By the partial information setting this implies that the agent is part of $\Sset$ in its group or pod. If it remains selected when it reports its true marginal utility then it receives the same reward price. So it is no worse off being truthful.

Suppose agent $k$ is not selected. If it remains not selected when it reports its true marginal utility, it faces the same penalty price. So again, agent $k$ is no worse off being truthful.

The only remaining situations to consider are when agent $k$ strategically reports its marginal utility to alter the selection decision:
\setlength{\leftmargini}{0.3in}
\begin{itemize}
\item {\em Agent $k$ is selected, but would not selected if it is truthful.} \\
Since agent $k$ is selected, from Proposition \ref{prop:2}, its optimal cost is either $J_1 = -\pi^r_k b_k$ or $J_1 = -(\pi_k - \pi^e) b_k$ depending on whether $\pi^r_k \geq \pi_k - \pi^e$ or $\pi^r_k \leq \pi_k - \pi^e$ resply.
The reward price received by agent $k$ under SRBM is
\[ \pi^r_k =  \max\{\mu_j\} -\pi^e \ , \ j \in {\Sset}_{-k} \]

If agent $k$ is not selected when it is truthful, using Proposition \ref{prop:2}, its optimal cost is $J_2 = -(\pi_k - \pi^e)b_k$.
Since agent $k$ is not selected, we can drop $k$ from the list of agents
in determining the set of selected agents. Thus, the set of selected agents in this sub-case is simply $\Sset_{-k}$, and since agent $k$ is not selected,
we have $\pi_k > \max_{j \in \Sset_{-k}} \ \mu_k = \pi^r_k + \pi^e$. As a result, $\pi^r_k < \pi_k - \pi^e$, which implies $J_1 \geq J_2$. Thus agent
$k$ is better off being truthful. 

\item {\em Agent $k$ is not selected, but would be selected if it is truthful.} \\
Since agent $k$ is not selected, using Proposition \ref{prop:2}, its optimal cost is $J_1 = -(\pi_k - \pi^e)b_k$.

If agent $k$ is selected when it is truthful, using Proposition \ref{prop:2}, its optimal cost is $J_2 = -\pi^r_k b_k$ or $J_2 = -(\pi_k - \pi^e) b_k$.
Also, 
\[ \pi_k^r =  \max\{\mu_j\} -\pi^e \ , \ j \in {\Sset}_{-k} \]
Since agent $k$ is selected when truthful, we have $\pi_k <  \max\{\mu_j: \ j \in {\Sset}_{-k}\}$. As a result,
$\pi_k^r > \pi_k - \pi^e$, which implies $J_1 \geq J_2$. Thus agent
$k$ is better off being truthful. 

\end{itemize}

In all situations, agent $k$ is no worse off by truthfully reporting its marginal utility and baseline, proving (a). \\
(b) and (c) now follow immediately from Proposition \ref{prop:2}. \\
(d) is a straightforward observations. \\
(e) Recall that the set of selected agents from  pod $\Pset^i$ (see (\ref{eq:select-m-muuc})) is $\Sset = \{1, \cdots, k^*\}$ where
\[
\sum_{k=1}^{k^*} f_k \geq D, \quad \sum_{k=1}^{k^*-1} f_k <  D 
\]
Note that (a) agents report baselines truthfully, and (c) selected agents completely yield their discretionary consumption.
As a result, the agents in pod core when selected provide $\sum_{k=1}^{k^*} b_k$ in demand reduction, meeting the target $D$.
\hfill $\Box$

\subsection{Proof of Theorem \ref{thm:MNbounds}} 

\label{sec:pfMNbounds}
%

{\em Proof:} 

(a) Let $N_i$ be the number of agents in the pod core $\Sset^i$.
Recall agents are sorted into $\Sset^i$ so that 
\[ \sum_{k \in \Sset^i} b_i \geq D. \]
Using Theorem \ref{thm:optional}, we can bound the expected number of agents in the pod core $\Sset^i$ as
\beq 
\label{eq:754} \frac{D}{\Eset [ b ]}  \leq \Eset [ N_i ] \leq \frac{D}{\Eset [ b ]} + 1 
\eeq
for all $i$. Note that this is also a bound on the number of agents in any pod header $\Hset^i$.
Next observe from Figure \ref{fig:pods} that the header for pod $\Pset^i$ serves as the core $\Sset^{i+1}$ for the subsequent pod.
Thus the total number of agents recruited is 
\[ N = \sum_1^{M+1} N_i. \]
It is easy to observe that $N_{i}$s are i.i.d. The number of pods $M$ is determined  using the stopping criterion $\sum_1^M \beta^i \geq 1$ and hence $M$ is a stopping time. Using Wald's equation, from the above equation we get, 
\beq \label{eq:9992} \Eset[N] =  (\Eset[ M ] + 1) \Eset[N_1]  \leq (\Eset[ M ] + 1) (D/\Eset [ b ] + 1)   \eeq

%
%

(b) We now bound the number of pods $M$. Recall that we keep forming pods until
\beq \label{eq:300}
\sum_1^{M-1}  \beta^i  < 1 \leq  \sum_1^{M} \beta^i, \ \text{where } \beta^i = \frac{\pi^e}{\nu^i}
\eeq
Since the harmonic mean of a collection of numbers is always less than arithmetic mean, we have
\[ \frac{M-1}{1/\nu^1 + \cdots + 1/\nu^{M-1}} \leq \frac{\nu^1 + \cdots + \nu^{M-1}}{M-1} \]
As a result, we have
\begin{eqnarray*}
 (M-1)^2  & \leq & \pi^e(1/\nu^1 + \cdots + 1/\nu^{M-1})\cdot \frac{\nu^1 + \cdots + \nu^{M-1}}{\pi^e}  \\
 & < &  \frac{\nu^1 + \cdots + \nu^{M-1}}{\pi^e} 
 \end{eqnarray*}
The final inequality above follows from (\ref{eq:300}). Note that $(M-1)^2$ is a convex function of $M$. Taking expectations,  using Jensen's inequality, we obtain
\beq 
\label{eq:301}
(\Eset [ M-1 ])^{2} \leq \Eset [ (M-1)^2 ] <  \frac{1}{\pi^e} \Eset \left[ \sum_{i=1}^{M-1}  \nu^i \right] 
\eeq

Define $\Sset^{M+1} = \Hset^M$. Agents in this  $(M+1)^\text{th}$ pod core are never
called. They are only needed to serve as the pod header $\Hset^M$ in order to determine reward prices for pod $\Pset^M$ under SRBM. 
Define $\pi_{[k]}, k = 1,2, \cdots$ to be the agent marginal utilities {\em sorted in increasing order}. Notice that (see Figure \ref{fig:pods})
\[ \nu^i  \leq \pi_{[k]}  \quad \text{for } k \in \Sset^{i+2}. \]
Therefore
\[ N_{i+2} \nu^i \leq \sum_{k \in \Sset^{i+2}} \pi_{[k]} \]
where $N_i$ is the number of agents in pod core $\Sset^i$.

As a result, we have
\begin{align} 
\nonumber
&\sum_{i=1}^{M-1} N_{i+2} \nu^i \leq  \sum_{i = 1}^{M-1} \sum_{k \in \Sset^{i+2}} \pi_{[k]} = \sum_{k \in \Sset^{3} \cup \cdots \cup \Sset^{M+1}} \pi_{[k]} \nonumber \\ 
\label{eq:302}
& \leq  \sum_{k \in \Sset^{1} \cup \cdots \cup \Sset^{M+1}} \pi_{[k]}  =  \sum_{k=1}^{N} \pi_{[k]}  =  \sum_{k=1}^{N} \pi_{k}
\end{align}
The final equality follows because the sum of all of the {\em sorted} random variables $\pi_{[k]}$ is identical to the 
sum of the unsorted random variables $\pi_k$. Taking  expectations of both sides of (\ref{eq:302}), 
\begin{align}
\mathbb{E}\left[ \sum_{i=1}^{M-1} N_{i+2} \nu^i \right] &\leq \mathbb{E}\left[ \sum_{k=1}^{N} \pi_{k} \right] = \mathbb{E}[N] \mathbb{E}[\pi] \nonumber  \\
\label{eq:dkop3}
&=  \mathbb{E}[N_{1}] \mathbb{E}[(M+1)] \Eset[ \pi ].
\end{align} 
where we used \eqref{eq:9992} to get the last equality. Notice that $N_{i+2}$ and $\nu^{i}$ are independent, and since $N_{i}$s are i.i.d.,  by a simple conditioning argument,  
\begin{align}
\label{eq:dkop4}
\mathbb{E}\left[ \sum_{i=1}^{M-1} N_{i+2} \nu^i \right] = \mathbb{E}[N_{1}] \mathbb{E}\left[ \sum_{i=1}^{M-1}  \nu^i \right]
\end{align}
From \eqref{eq:dkop4} and \eqref{eq:dkop3}, 
\begin{align}
\label{eq:dkop5}
\mathbb{E}\left[ \sum_{i=1}^{M-1}  \nu^i \right] \leq  \mathbb{E}[(M+1)] \Eset[ \pi ]
\end{align}
Combining  \eqref{eq:301} and (\ref{eq:dkop5}), we obtain
\[
\left(\Eset[ M ]  + 1 \right) \left(\Eset [ M ] - 3] \right) \leq (\mathbb{E}[(M-1)])^{2} \leq \frac{\Eset [ M+1 ] \Eset [ \pi ]}{\pi^e}
\]
This simplifies to
\[   \Eset [ M ] < \frac{\Eset [ \pi ]}{\pi^e} + 3 \]
proving the claim. \hfill $\Box$

\subsection{Proof of Theorem \ref{thm:cc}} 
Agent $k$ in Pod $\Pset^i$ is selected with probability $\beta^i_k  = \pi^e/(\pi^r_k + \pi^e)$. Selected agents deliver $D$ KWh of demand reduction. Thus, the expected payout to agents per DR event is
\begin{align}
&\hspace{-0.4cm}\Eset[\psi] \leq \Eset [ \sum_{i=1}^M \sum_{k \in \Sset^i} \beta^i_k \pi^r_k b_k ]  \leq  \Eset[\sum_{i=1}^M \pi^e D ] \nonumber \\
\label{eq:dk912}
&  =  \Eset[\sum_{i=1}^M \pi^e D] \leq E[M] \pi^e D.
\end{align}
From (\ref{eq:aggcost}), the average cost of DR provision per KWh of demand reduction is 
\[
\phi =  \frac{\Eset[ \psi ]}{D} + \frac{\pi^o \Eset[N]}{mD} 
\]
Using the  bound on $\Eset[\psi]$ from \eqref{eq:dk912} and the bound on $\Eset[N]$ from Theorem \ref{thm:MNbounds}, we get the desired result. \hfill $\Box$

\subsection{SRBM for Complete Information (CI) Setting}
\label{sec:srbm-ci}

The main challenge in the mechanism design is that the selection probabilities of individual consumers has to be limited (see Theorem \ref{thm:1}) so that the reported baselines are not inflated. This is unlike traditional resource allocation problems where the baseline is available. This necessitates the mechanism designer to recruit more than the minimum number of agents required to meet the target $D$. SRBM recruits multiple such sets of consumers and sorts them in to pods, where each pod has a core that contains the minimum number of agents required to meet the target $D$. Supposing that the agents reveal the true values of their baseline and marginal utility, then the mechanism can minimize its payout by selecting pod cores of higher marginal utilities by a lower proability. The reward design (\eqref{eq:pslope-main}), would then imply that pod cores with higher rewards are selected less often.

In the complete information setting this can couple the selection probabilities with the agent reporting. As a result,  a characterization like Thm. \ref{thm:1} might not be feasible. So the challenge here is to design a {\it pod sorting} algorithm, combined with a {\it pod selection} and {\it reward pricing} mechanism such that the incentive for the agents to inflate their baseline reports or their marginal utility reports to increase payments is nullified while pod sorting is cost effective.

Below we discuss the mechanism for this complete information setting. Agents are sorted in to pods based on their reported marginal utility as in Fig. \ref{fig:pods}. Agent $k$'s selection probability and unit reward is given by,

\beq \beta^i_{k} = c^i_k \pi^e/ \nu^{i+1}_{-k} ; \ \pi^r_k = \nu^{i}_{-k} - \pi^e \eeq

Where,
\begin{align}
& \mu^{j+}_k = \nu^{j}_{-k} \frac{\nu^{j+2}_{-k} - e^j_k \left(\nu^{j+1}_{-k}\right)^2/\nu^{j}_{-k}}{\nu^{j+2}_{-k} - e^j_k\nu^{j+1}_{-k}}\ , \ c^1_k = 1 \ , \ e^j_k = \frac{c^{j+1}_k}{c^j_k} \nonumber\\
& \mu^{j+}_k =  \nu^{i-1}_{-k} \ , \ \text{if} \ j > i \ , \  \mu^{j+}_k =  \nu^{j-1}_{-k} \ \text{if} \ j \leq i \ , \nonumber\\
& \text{And} \ \nu^{j}_{-k} = \text{highest report in} \ j \ \text{th pod if} \ k \ \text{was excluded} \nonumber\\
& \text{And} \ i \ \text{is the pod allotted to consumer} \ k 
\label{eq:jumpfact} 
\end{align}

The above equation can be rewritten as,
\begin{equation}
e^j_k = \frac{\nu^{j+2}_{-k}(\mu^{j+}_{k} - \nu^j_{-k})}{\nu^{j+1}_{-k}(\mu^{j+}_{k} - \nu^{j+1}_{-k})} > 0
\label{eq:jumpfact-1} 
\end{equation}

\begin{remark}
There are three differences from mechanism for PI setting, the reward, the selection probability and the inclusion of an additional factor in the calculation of the selection probability.
\end{remark}

The number of pods $M$ in the mechanism is chosen such that  $\sum_{j =1}^{M-1} \beta^{j} < 1 \leq \sum_{j =1}^{M} \beta^{j} $, where a $\Pset^i$'s core $\Sset^i$ is such that $\sum_{k = 1}^{N-1} f_k <  D \leq \sum_{k = 1}^{N} f_k$ and $N$ is the number of consumers in $\Sset^i$. Here the difference is that  $\beta^i \neq \min \beta_k$. {\it Selection:} a random number $u \sim \mathcal{U}[0,1]$ is drawn. Pod core $\Sset^i$ is selected for DR if $ u \in [ \sum_{j =1}^{i-1} \beta^j, \sum_{j =1}^{i} \beta^j] $. To be consistent with the selection probability of an individual consumer a consumer $k$ in $\Sset^i$ is selected for all draws $ u \in [ \sum_{j =1}^{i-1} \beta^j, \sum_{j =1}^{i-1} \beta^j + \beta_k] $. Agents who are selected are rewarded for reduction from reported baseline based on unit reward announced to them. Agents who are not selected are penalized for consuming less than the reported baseline. \\

{\it \bf Difference between SRBM for CI and SRBM for PI}: Unlike SRBM for PI, the selection probability of agent $k$ has an additional jump factor $c^i_k$. This guarantees that the selection probability of an individual agent decays at a sufficient rate that the agent does not have the incentive to misreport and jump to a higher pod just to gain higher reward. Also note that the reward and the probability without the factor are also different. In SRBM for PI, the reward is $\pi^r_k = \nu^{i-1}_{-k} - \pi^e$ while in SRBM for CI the reward is $\pi^r_k = \nu^{i}_{-k} - \pi^e$. Similarly the selection probability in SRBM for PI is $\beta^i_k = \pi^e/\nu^i_{-k}$ while in SRBM for CI the selection probability without the jump factor is $\beta^i_k/c^i_k = \pi^e/\nu^{i+1}_{-k}$. This change in the selection probability and the reward enables design of a mechanism where the jump factors are computable using the reported information.  

\begin{theorem}\label{thm:muuc-main-srbm-ci}
The modified SRBM  for the complete information setting has the following properties,
\vspace*{-0.08in}
\setlength{\leftmargini}{0.15in}
\balphlist \setlength{\itemsep}{-0.1cm}
\item truthful reporting of baselines and marginal utilities 
is a dominant strategy 
\item called agents consume $q_\text{s}^* = 0$, providing the maximal reduction in their discretionary consumption
\item agents that are not called consume $q_\text{ns}^* = b_k$
\item the aggregator receives no penalty revenue 
\item the load reduction target $D$ is met by each pod core.
\end{list}
\end{theorem}

\begin{proof} The jump factor $c^i_k$ (\eqref{eq:jumpfact}), that is used in the selection probability of agent $k$,  is only dependent on $\mu^{(i-1)+}_k, \nu^{i-1}_{-k}, \nu^{i+1}_{-k}, \nu^{i}_{-k}$ and $c^{i-1}_k$. Hence the recursive scheme for computation of $c^i_k$ is independent of the agent report $\mu_k$ and that the jump factor $c^i_k$ is independent of the reports of agent $k$. This implies that the selection probability of a agent $k, \beta_k = c^i_k \pi^e/\nu^{i+1}_{-k}$ is independent of its reports. \\

Now all we need to check is that whether $\beta_k \leq \pi^e/(\pi^r_k + \pi^e)$. When $e^j_k$ is zero on the right of \eqref{eq:jumpfact}, we get,
\beq
\mu^{j+}_k \leq \nu^{j-1}_{-k} < \left. \nu^{j}_{-k} \frac{\nu^{j+2}_{-k} - e^j_k \left(\nu^{j+1}_{-k}\right)^2/\nu^{j}_{-k}}{\nu^{j+2}_{-k} - e^j_k\nu^{j+1}_{-k}}\right\vert_{e^j_k = 0} = \nu^{j}_{-k}
\eeq
Calling right hand side of  \eqref{eq:jumpfact} as $g(\nu^j_{-k}, \nu^{j+1}_{-k}, \nu^{j+2}_{-k}, e^j_k)$, it is easy to check using chain rule for differentiation that $\frac{d g}{d e^j_k} < 0$. And $g(\nu^j_{-k}, \nu^{j+1}_{-k}, \nu^{j+2}_{-k}, e^j_k) = 0$ when $c^{j+1}_k = \frac{c^j_k \nu^{j+2}_{-k} \nu^{j}_{-k}}{(\nu^{j+1}_{-k})^2}$. Combining these two observations it is clear that there is a $c^i_k$ such that,
\beq
0 < c^i_k < \frac{c^{i-1}_k \nu^{i+1}_{-k} \nu^{i-1}_{-k}}{(\nu^{i}_{-k})^2}
\eeq
Similarly it follows that,
\beq
0 < c^{i-1}_k < \frac{c^{i-2}_k \nu^{i}_{-k} \nu^{i-2}_{-k}}{(\nu^{i-1}_{-k})^2}
\eeq
Applying recursively we get that,
\beq
0 < c^i_k < \frac{ \nu^1_{-k} \nu^{i+1}_{-k} }{ \nu^2_{-k} \nu^{i}_{-k} } < \frac{ \nu^{i+1}_{-k} }{ \nu^{i}_{-k}} 
\eeq

This implies the selection probability of agent $k$ satisfies,
\beq
\beta_k = c^i_k \pi^e/\nu^{i+1}_{-k} < \frac{ \nu^{i+1}_{-k} }{ \nu^{i}_{-k}}  \pi^e/\nu^{i+1}_{-k}  = \frac{\pi^e}{\nu^{i}_{-k}} = \frac{\pi^e}{\pi^r_k + \pi^e}
\eeq
From the above observation and the fact that $\beta_k$ and $\pi^r_k$ are independent of agent $k$'s reports. Then the proof of Thm. \ref{thm:1} applies here and it follows that the agent $k$ reports its {\it true baseline} i.e. $f_k = b_k$. \\

Note that it is trivial to show that the agent does not gain by changing its marginal utility report such  that it is still within $\Sset^i$ (pod core allotted if agent $k$ reports truthfully) but at a different position. This is because the rewards, selection probabilities for a agent $k$ does not change once its pod core is fixed. Next we show that the agent does not gain either by jumping to a lower or a higher pod core. \\

Denote by $i$ the pod core number that is allotted to agent $k$ on reporting truthfully and denote by $U^i_k$ the net utility of agent $k$ in $\Sset^i$. Then,
\begin{align}
U^i_k & = \beta^i_k \pi^r_k b_k + (1 -  \beta^i_k) (\pi_k - \pi^e)b_k = \frac{c^i_k\pi^e}{\nu^{i+1}_{-k}} (\nu^i_{-k} - \pi^e) b_k \nonumber\\
& + \left(1 -  \frac{c^i_k\pi^e}{\nu^{i+1}_{-k}}\right) (\pi_k - \pi^e)b_k
\label{eq:pf-1} 
\end{align}

Define $\tilde{U}^j_k$ as the utility of agent $k$ in $\Sset^j$, but using the computed jump factor $c^j_k$ when $k$ is allotted $\Sset^i$. Then,
\begin{align}
\tilde{U}^{j}_k & = \beta^{j}_k \pi^r_k b_k + (1 -  \beta^{j}_k) (\pi_k - \pi^e)b_k \nonumber\\
& = \frac{c^{j}_k\pi^e}{\nu^{j+1}_{-k}} (\nu^{j}_{-k} - \pi^e) b_k  + \left(1 -  \frac{c^{j}_k\pi^e}{\nu^{j+1}_{-k}}\right) (\pi_k - \pi^e)b_k
\label{eq:pf-2} 
\end{align}
 
{\it Case $j \geq i$:} Here we first show that $\tilde{U}^{j+1}_k \leq \tilde{U}^{j}_k$ and then show that $U^j_k \leq \tilde{U}^{j}_k \leq \tilde{U}^{i}_k = U^i_k$. Taking the difference of $\tilde{U}^{j+1}_k$ and $\tilde{U}^{j}_k$ we get,
\begin{align} 
& \tilde{U}^{j+1}_k - \tilde{U}^{j}_k = \frac{c^{j+1}_k\pi^e}{\nu^{j+2}_{-k}} \nu^{j+1}_{-k}  b_k + \left(1 -  \frac{c^{j+1}_k\pi^e}{\nu^{j+2}_{-k}}\right) \pi_k b_k  \nonumber\\
& - \frac{c^{j}_k\pi^e}{\nu^{j+1}_{-k}} \nu^{j}_{-k}  b_k - \left(1 -  \frac{c^{j}_k\pi^e}{\nu^{j+1}_{-k}}\right) \pi_k b_k \nonumber\\
& = \frac{c^{j+1}_k\pi^e}{\nu^{j+2}_{-k}} \nu^{j+1}_{-k}  b_k - \frac{c^{j+1}_k\pi^e}{\nu^{j+2}_{-k}} \pi_k b_k  - \frac{c^{j}_k\pi^e}{\nu^{j+1}_{-k}} \nu^{j}_{-k}  b_k \nonumber\\
& + \frac{c^{j}_k\pi^e}{\nu^{j+1}_{-k}} \pi_k b_k \nonumber\\
& = \frac{c^j_k \pi^e b_k}{\nu^{j+2}_{-k} \nu^{j+1}_{-k}}\left( e^j_k (\nu^{j+1}_{-k})^2 - \nu^{j+2}_{-k} \nu^{j}_{-k}  -  \pi_k (e^j_k \nu^{j+1}_{-k} - \nu^{j+2}_{-k} ) \right) \nonumber\\
& = \frac{c^j_k \pi^e b_k\pi_k}{\nu^{j+2}_{-k} \nu^{j+1}_{-k}}\left( \pi_k \left(\nu^{j+2}_{-k} - e^j_k \nu^{j+1}_{-k}\right) \right. \nonumber\\
& \left. -  \nu^{j}_{-k} \left(\nu^{j+2}_{-k} - e^j_k \frac{(\nu^{j+1}_{-k})^2}{\nu^{j}_{-k}} \right) \right)
\label{eq:pf-2} 
\end{align}

Then from  \eqref{eq:jumpfact}, \eqref{eq:pf-2} and the fact that $\mu^{j+}_k = \nu^{i-1}_{-k} \geq \pi_k$, because agent $k$'s true marginal utility $\pi_k$ satisfies $\nu^{i-2}_{-k} \leq \pi_k \leq \nu^{i-1}_{-k}$, we get,

\begin{align}
&\tilde{U}^{j+1}_k - \tilde{U}^{j}_k \leq \frac{c^j_k \pi^e b_k\pi_k}{\nu^{j+2}_{-k} \nu^{j+1}_{-k}}\left( \mu^{j+}_{-k} \left(\nu^{j+2}_{-k} - e^j_k \nu^{j+1}_{-k}\right) \right. \nonumber\\
& - \left. \nu^{j}_{-k} \left(\nu^{j+2}_{-k} - e^j_k \frac{(\nu^{j+1}_{-k})^2}{\nu^{j}_{-k}} \right) \right) = 0 \nonumber\\
& \Rightarrow \tilde{U}^{j+1}_k \leq \tilde{U}^{j}_k
\end{align}

Next we show that $c^j_{k \in \Sset^r} \leq c^j_{k \in \Sset^i}$, when $r > i$, and then it trivally follows that $U^j_k \leq \tilde{U}^{j}_k$. If the agent $k$ jumps to higher pod core $r$ (by reporting high) then the value $\mu^{j+}_{k \in \Sset^r} = \min\{\nu^{j-1}_{-k},\nu^{r-1}_{-k}\} \geq \nu^{i-1}_{-k} = \mu^{j+}_{k \in \Sset^i} \ \forall \ j > i$ and $\mu^{j+}_{k \in \Sset^r} = \nu^{j-1}_{-k} = \mu^{j+}_{k \in \Sset^i} \ \forall \ j \leq i$ . Also, \eqref{eq:jumpfact} becomes $\mu^{(j-1)+}_{k \in \Sset^r} = g(\nu^{j-1}_{-k}, \nu^{j}_{-k}, \nu^{j+1}_{-k}, e^{j-1}_{k \in \Sset^r})$. For a particular $j > i$, every argument of $g$ is fixed except $e^j_k$, which depends on the core $r$ agent $k$ is in. We showed before that the right hand side of \eqref{eq:jumpfact} i.e. $g(\nu^j_{-k}, \nu^{j+1}_{-k}, \nu^{j+2}_{-k}, e^j_k)$ is decreasing in $e^j_k$. Then it follows that $e^j_{k \in \Sset^r} \leq e^j_{k \in \Sset^i} \ r > i$. Since $c^1_k = 1$ always it follows that $c^j_{k \in \Sset^r} \leq c^j_{k \in \Sset^i}$ when $r > i$. Setting $j = r$ it follows that $c^j_{k \in \Sset^j} \leq c^j_{k \in \Sset^i}$. This implies that $U^j_k \leq \tilde{U}^{j}_k$. From here it follows that $U^j_k \leq \tilde{U}^{j}_k \leq \tilde{U}^{i}_k = U^i_k$. \\

{\it Case $j < i$:} First the agent does not have any incentive to move to a pod core $\Sset^j$ where $j \leq i-2$. This is because the reward for agent $k$ when it is in a core $j \leq i-2$ is $\pi^r_k = \nu^j_{-k} \leq \nu^{i-2}_{-k} \leq \pi_k$. This implies that the incentive to reduce is not sufficient in any of these pod cores. Hence the agent's utility $U^j_k = (\pi_k - \pi^e)b_k < U^i_k$. From \eqref{eq:jumpfact} it follows that $c^j_{k \in \Sset^j} = c^j_{k \in \Sset^i}$ and $c^i_{k \in \Sset^j} = c^i_{k \in \Sset^i}$ when $j = i-1$. This implies $U^j_k = \tilde{U}^j_k = \tilde{U}^{i-1}_k$. Finding the difference between $\tilde{U}^{i}_k$ and $\tilde{U}^{i-1}_k$ we get,
\begin{align} 
& \tilde{U}^{i}_k - \tilde{U}^{i-1}_k = \frac{c^{i-1}_k \pi^e b_k\pi_k}{\nu^{i+1}_{-k} \nu^{i}_{-k}}\left( \pi_k \left(\nu^{i+1}_{-k} - e^{i-1}_k \nu^{i}_{-k}\right) \right. \nonumber\\
& -  \left. \nu^{i-1}_{-k} \left(\nu^{i+1}_{-k} - e^{i-1}_k \frac{(\nu^{i}_{-k})^2}{\nu^{i-1}_{-k}} \right) \right) \nonumber\\
& \geq \frac{c^{i-1}_k \pi^e b_k\pi_k}{\nu^{i+1}_{-k} \nu^{i}_{-k}}\left( \nu^{i-2}_{-k} \left(\nu^{i+1}_{-k} - e^{i-1}_k \nu^{i}_{-k}\right)  \right. \nonumber\\
& -  \left. \nu^{i-1}_{-k} \left(\nu^{i+1}_{-k} - e^{i-1}_k \frac{(\nu^{i}_{-k})^2}{\nu^{i-1}_{-k}} \right) \right)
\label{eq:pf-3} 
\end{align}

Then from \eqref{eq:jumpfact} it follows that,
\begin{align} 
& \tilde{U}^{i}_k - \tilde{U}^{i-1}_k \geq \frac{c^{i-1}_k \pi^e b_k\pi_k}{\nu^{i+1}_{-k} \nu^{i}_{-k}}\left( \nu^{i-2}_{-k} \left(\nu^{i+1}_{-k} - e^{i-1}_k \nu^{i}_{-k}\right) \right.  \nonumber\\
& - \left. \nu^{i-1}_{-k} \left(\nu^{i+1}_{-k} - e^{i-1}_k \frac{(\nu^{i}_{-k})^2}{\nu^{i-1}_{-k}} \right) \right) \nonumber\\
& = 0 \Rightarrow \tilde{U}^{i}_k \geq  \tilde{U}^{i-1}_k \Rightarrow U^{i}_k \geq U^{i-1}_k = U^{j}_k
\label{eq:pf-4} 
\end{align}

This establishes the first point. The other points follow from here. Hence proved. \end{proof}

\begin{IEEEbiography}{Deepan Muthirayan}
is currently a Post-doctoral Researcher in the department of Electrical Engineering and Computer Engineering at UC Irvine. He obtained his Phd in Mechanical Engineering from the University of California at Berkeley (2016) and his B.Tech/M.tech from the Indian Institute of Technology Madras (2010). His thesis work focussed on market mechanisms for integrating demand flexibility in energy systems. His research lies in control theory, machine learning, learning for control, online learning, online algorithms, game theory, smart systems.
\end{IEEEbiography}

\begin{IEEEbiography}{Dileep  Kalathil}
is an Assistant  Professor in the Department of Electrical and Computer Engineering  at Texas A\&M  University, College Station, Texas. He was a postdoctoral researcher in the EECS department, University of California, Berkeley from 2014-17. He received his PhD from University of Southern California (USC) in 2014 where he won the best PhD Dissertation Prize in the USC Department of Electrical Engineering. He received an M.Tech from IIT Madras where he won the award for the best academic performance in the EE department. His research interests include  intelligent transportation systems, sustainable energy systems, data driven optimization, online learning, stochastic control, and game theory.
\end{IEEEbiography}

\begin{IEEEbiography}
{Kameshwar Poolla}
 is the Cadence Distinguished Professor at UC Berkeley in EECS and ME. His current research interests include many aspects of future energy systems including economics, security, and commercialization. He was the Founding Director of the IMPACT Center for Integrated Circuit manufacturing. Dr. Poolla co-founded OnWafer Technologies which was acquired by KLA-Tencor in 2007. Dr. Poolla has been awarded a 1988 NSF Presidential Young Investigator Award, the 1993 Hugo Schuck Best Paper Prize, the 1994 Donald P. Eckman Award, the 1998 Distinguished Teaching Award of the University of California, the 2005 and 2007 IEEE Transactions on Semiconductor Manufacturing Best Paper Prizes, and the 2009 IEEE CSS Transition to Practice Award.
\end{IEEEbiography}

\begin{IEEEbiography}
{Pravin Varaiya} is a Professor of the Graduate School in the Department of Electrical Engineering and Computer Sciences at the University of California, Berkeley.  He has been a Visiting Professor at the Institute for Advanced Study at the Hong Kong University of Science and Technology since 2010.  He has co-authored four books and 350+ articles.  His current research is devoted to electric energy systems and transportation networks.   

Varaiya has held a Guggenheim Fellowship and a Miller Research Professorship.  He has received three honorary doctorates, the Richard E. Bellman Control Heritage Award, the Field Medal and Bode Lecture Prize of the IEEE Control Systems Society, and the Outstanding Researcher Award of the IEEE Intelligent Transportation Systems Society. He is a Fellow of IEEE, a Fellow of IFAC, a member of the National Academy of Engineering, and a Fellow of the American Academy of Arts and Sciences.
\end{IEEEbiography}

\end{document}